\documentclass[a4paper,11pt]{article}
\usepackage[margin=1in]{geometry}
\usepackage[utf8]{inputenc}

\usepackage{fontaxes}
\usepackage{trimspaces}
\usepackage{nccfoots}
\usepackage{setspace}
\usepackage{inconsolata}
\usepackage{libertine}
\usepackage{thm-restate}

\usepackage[T1]{fontenc}
\usepackage{amsmath, amssymb}
\usepackage{xargs}
\usepackage{mathtools}
\usepackage{amsthm}
\usepackage{stmaryrd}
\usepackage[sort&compress,numbers,square,comma,longnamesfirst]{natbib}
\usepackage{graphicx}
\usepackage{hyperref}
\usepackage{xcolor}
\usepackage{comment}
\usepackage{todonotes}
\usepackage[linesnumbered,lined,commentsnumbered,noend,boxed]{algorithm2e}
\usepackage{enumitem}
\usepackage{framed}
\usepackage[noabbrev,capitalise]{cleveref}
\usepackage[framemethod=TikZ]{mdframed}

\usepackage{xspace}
\usepackage{subcaption}

\numberwithin{equation}{section}
\numberwithin{figure}{section}

\newtheorem{theorem}{Theorem}[section]
\newtheorem{definition}[theorem]{Definition}
\newtheorem{lemma}[theorem]{Lemma}
\newtheorem{corollary}[theorem]{Corollary}
\newtheorem{claim}[theorem]{Claim}
\newtheorem{proposition}[theorem]{Proposition}

\newtheorem{fact}[theorem]{Fact}

\newtheorem{property}[theorem]{Property}
\newtheorem{hypothesis}[theorem]{Hypothesis}

\crefname{obs}{Observation}{Observations}
\Crefname{obs}{Observation}{Observations}
\crefname{fact}{Fact}{Facts}
\Crefname{fact}{Fact}{Facts}
\crefname{problem}{Problem}{Problems}
\Crefname{problem}{Problem}{Problems}
\crefname{conjecture}{Conjecture}{Conjectures}
\Crefname{conjecture}{Conjecture}{Conjectures}
\crefname{claim}{Claim}{Claims}
\Crefname{claim}{Claim}{Claims}

\newcommand{\cint}[1]{\lfloor #1 \rceil}

\newcommand{\eps}{\varepsilon}

\newcommand{\Oh}{\mathcal{O}}

\newcommand{\Otilde}{\widetilde{\mathcal{O}}}

\newcommand{\tOh}{\Otilde}

\newcommandx{\unsure}[2][1=]{\todo[linecolor=green,backgroundcolor=green!25,bordercolor=green,#1]{\normalsize #2}}
\newcommandx{\improvement}[2][1=]{\todo[inline,linecolor=blue,backgroundcolor=blue!05,bordercolor=blue,#1]{\normalsize #2}}
\newcommandx{\info}[2][1=]{\todo[linecolor=yellow,backgroundcolor=yellow!25,bordercolor=yellow,#1]{#2}}
\newcommandx{\floatmodel}[2][1=]{\todo[inline,linecolor=red,backgroundcolor=yellow!25,bordercolor=yellow,#1]{#2}}
\newcommandx{\thiswillnotshow}[2][1=]{\todo[disable,#1]{#2}}

\DeclarePairedDelimiter{\abs}{\lvert}{\rvert}

\def\epsilon{\varepsilon}
\def\eps{\epsilon}

\newcommand{\R}{\ensuremath{\mathbb{R}}\xspace}

\newcommand{\etal}{et al.\xspace}

\newcommand{\Red}{R}
\newcommand{\Blue}{B}
\newcommand{\D}{\ensuremath{\mathcal{D}}\xspace}
\newcommand{\A}{\ensuremath{\mathcal{A}}\xspace}
\newcommand{\EMD}{\ensuremath{\mathrm{EMD}}\xspace}
\newcommand{\Cost}{\mathrm{Cost}}
\newcommand{\EMDuT}{\ensuremath{\mathrm{EMDuT}}\xspace}

\graphicspath{{figures/}}

\title{Fine-Grained Complexity of Earth Mover's Distance under
Translation\thanks{This work was initiated at the Workshop on New Directions in
Geometric Algorithms, May 14-19 2023, Utrecht, The Netherlands.}}

\author{
  Karl Bringmann\thanks{Saarland University and Max-Planck-Institute for Informatics, Saarland Informatics Campus,
        Saarbr\"ucken, Germany, \textsf{bringmann@cs.uni-saarland.de}. This work is part of the project TIPEA that has received funding from the European Research Council (ERC) under the European Unions Horizon 2020 research and innovation programme (grant agreement No.\ 850979).
  } \and
  Frank Staals\thanks{Department of Information and Computing
    Sciences, Utrecht University, The Netherlands, \textsf{f.staals@uu.nl}}
  \and
Karol W\k{e}grzycki\thanks{Max Planck Institute for Informatics, Saarland Informatics Campus,
      Saarbr\"ucken, Germany, \textsf{kwegrzyc@mpi-inf-mpg.de}. Supported by the Deutsche
Forschungsgemeinschaft (DFG, German Research Foundation) grant number
559177164.}
  \and
  Geert van Wordragen\thanks{Department of Computer Science, Aalto
    University, Espoo, Finland, \textsf{Geert.vanWordragen@aalto.fi}}
}

\date{}
\begin{document}

\maketitle

\begin{abstract}
The Earth Mover's Distance is a popular similarity measure in several branches
of computer science. It measures the minimum total edge length of a perfect
matching between two point sets.  The Earth Mover's Distance under Translation
($\mathrm{EMDuT}$) is a translation-invariant version thereof. It minimizes the Earth
Mover's Distance over all translations of one point set.

For $\mathrm{EMDuT}$ in $\mathbb{R}^1$, we present an $\widetilde{\mathcal{O}}(n^2)$-time
algorithm. We also show that this algorithm is nearly optimal by presenting a
matching conditional lower bound based on the Orthogonal Vectors Hypothesis. For
$\mathrm{EMDuT}$ in $\mathbb{R}^d$, we present an $\widetilde{\mathcal{O}}(n^{2d+2})$-time
algorithm for the $L_1$ and $L_\infty$ metric. We show that this dependence on
$d$ is asymptotically tight, as an $n^{o(d)}$-time algorithm for $L_1$ or
$L_\infty$ would contradict the Exponential Time Hypothesis (ETH). Prior to our
work, only approximation algorithms were known for these problems.

\end{abstract}

\thispagestyle{empty}

 \begin{picture}(0,0)
 \put(452,-345)
 {\hbox{\includegraphics[width=40px]{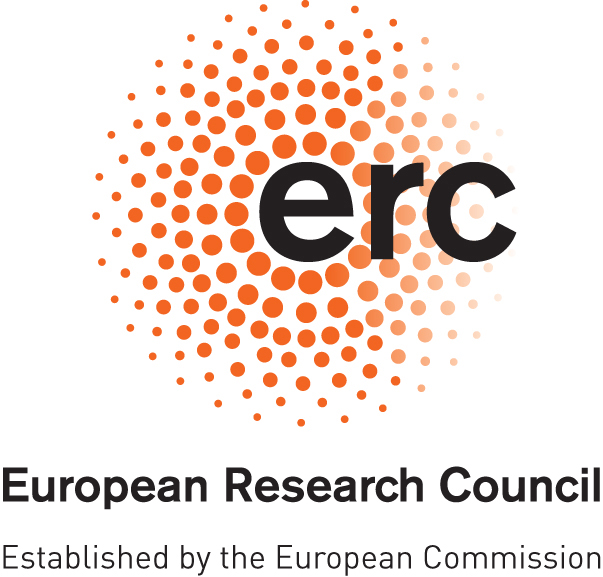}}}
 \put(442,-405)
 {\hbox{\includegraphics[width=60px]{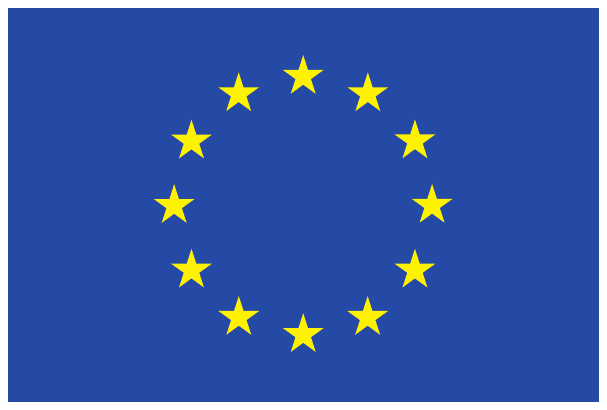}}}
 \end{picture}

\clearpage
\setcounter{page}{1}

\section{Introduction}

\paragraph{Earth Mover's Distance (EMD).}
\EMD, also known as geometric transportation or geometric bipartite
matching, is a widely studied distance measure (see,
e.g.,~\cite{Indyk07, AndoniIK08, AndoniBIW09, AndoniNOY14, Rohatgi19,
  KhesinNP20, FoxL22, FoxL23, AgarwalCRX22, AgarwalRSS22}) that
has received significant interest in computer vision, starting with
the work of~\cite{RubnerTG00}.  Depending on the precise formulation,
\EMD is a distance measure on point sets, distributions, or
functions. In this paper, we study the following formulation of \EMD
as measuring the distance from a set of blue points $B$ to a set of
red points $R$:
\begin{displaymath}
    \EMD_p(B,R) = \min_{\text{injective } \phi\colon B \to R} \sum_{b \in B} \|b - \phi(b)\|_p.
\end{displaymath}
Here, the minimization goes over all injective functions from $B$ to $R$, i.e., $\phi$ encodes a perfect matching of the points in $B$ to points in $R$, and the cost of a matching is the total length of all matching edges, with respect to the $L_p$ metric, $1 \le p \le \infty$. When the value of $p$ is irrelevant, we may drop the subscript~$p$.

The $\EMD_p$ problem is to compute the value $\EMD_p(B,R)$ for given sets $B,R \subseteq \R^d$ of sizes $|B| \le |R| = n$. This general problem is sometimes called the \emph{asymmetric} \EMD. The \emph{symmetric} \EMD is the special case with the additional restriction $|B| = |R|$. Intuitively, the asymmetric \EMD asks whether $B$ is similar to some subset of $R$, while the symmetric variant compares the full sets $B$ and~$R$. In this paper, we assume the dimension $d$ to be constant.

We briefly discuss algorithms for \EMD.
Note that \EMD can be formulated as a mincost matching problem on a bipartite
graph with vertices $R \cup B$, where edge lengths are equal to the
point-to-point distances. This graph has $|R|\cdot |B| = \Oh(n^2)$ edges and
solving bipartite mincost matching by the Hungarian method yields an exact
algorithm for \EMD with running time $\Oh(n^3)$.
Alternatively, by combining geometric spanners with recent advancements in
(approximate) mincost flow solvers, one can obtain fast approximation algorithms
for \EMD. For instance, symmetric \EMD in $L_2$ metric can be solved in time $n
(\log(n)/\eps)^{\Oh(d)}$~\cite{KhesinNP20}. See also~\cite{Indyk07, AndoniNOY14,
FoxL22, FoxL23, AgarwalCRX22, AgarwalRSS22} for more approximation
algorithms. Conditional lower bounds are also known, but
only when the dimension is super-constant~\cite{Rohatgi19}.

\paragraph{Earth Mover's Distance under Translation (EMDuT).}
We study a variant of \EMD that is invariant under translations, and thus compares shapes of point sets, ignoring their absolute positions:
\begin{displaymath}
 \EMDuT_p(B,R) = \min_{\tau \in \R^d} \EMD_p(B+\tau,R).
\end{displaymath}
Here, $B + \tau = \{b + \tau \mid b \in B\}$ is the translated point set. See Figure~\ref{fig:problem_definition} for an illustration of this distance measure.
Again, we call asymmetric $\EMDuT_p$ the problem of computing $\EMDuT_p(B,R)$ for given sets $B,R$ of size $|B| \le |R|=n$, and the symmetric variant comes with the additional restriction $|B| = |R|$.
This measure was introduced by Cohen and Guibas~\cite{CohenG99}, who
presented heuristics as well as an exact algorithm with respect to the squared Euclidean distance.
Later, Klein and Veltkamp~\cite{KleinV05} designed a 2-approximation
algorithm for symmetric $\EMDuT_p$ running in asymptotically the same
time as any \EMD algorithm. Cabello, Giannopoulos, Knauer, and
Rote~\cite{CabelloGKR08} designed $(1+\eps)$-approximation algorithms
for $\EMDuT_2$ in the plane, running in time $\tOh(n^4/\eps^4)$ for
the asymmetric variant and $\tOh(n^{3/2}/\eps^{7/2})$ for the
symmetric variant.\footnote{Here and throughout the paper we use $\tOh$ notation
to ignore logarithmic factors, i.e., $\tOh(T) = \bigcup_{c \ge 0} \Oh(T (\log T)^c)$.}
Eppstein~\etal~\cite{eppstein15improv_grid_map_layout_point_set_match}
proposed algorithms to
solve the symmetric $\EMDuT_1$ and symmetric $\EMDuT_\infty$ problems in the
plane, that run in $\Oh(n^6\log^3
n)$ time.
We remark that most of these works also study variants of \EMDuT under more general transformations than translations, but in this paper we focus on translations.

We are not aware of any other research on \EMDuT, which is surprising, since translation-invariant distance measures are well motivated, and the analogous Hausdorff distance under translation~\cite{HuttenlocherK90,Rote91,
HuttenlocherRK92,AgarwalHSW10,KnauerS11,KnauerKS11,
BringmannN21,Chan23} and Fr\'echet distance under translation~\cite{AltKW01,MosigC05,JiangXZ08,AvrahamKS15,
BringmannKN20,FiltserK20,BringmannKN21} have received considerably more attention.

\begin{figure}[tb]
  \centering
  \includegraphics{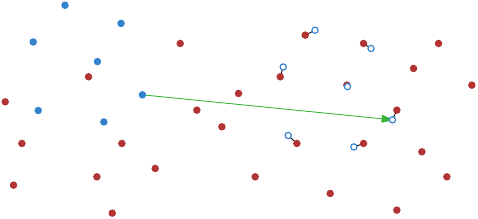}
  \caption{Given a set of (solid) blue points $B$ and a set of red points $R$, our goal is to find a translation~$\tau$ (shown in green) and a
    perfect matching from $B+\tau$ to $R$ (shown in black) that minimizes the total distance of matched pairs.
  }
  \label{fig:problem_definition}
\end{figure}

\subsection{Our Results}

We study \EMDuT from the perspective of fine-grained complexity. We design new algorithms and prove conditional lower bounds over $\R^1$, as well as for $L_1$ and $L_\infty$ over $\R^d$.

\paragraph{EMDuT in 1D.} Over $\R^1$ all $L_p$ metrics are equal. We present the following new algorithms.

\begin{theorem}[1D Algorithms] \label{thm:1Dalgo}
  (Symmetric:) Given sets $B,R \subseteq \R$ of size $n = |B| = |R|$,
  $\EMDuT(B,R)$ can be computed in time $\Oh(n \log n)$.
  (Asymmetric:) Given sets $B,R \subseteq \R$ of size $m = |B| \le n = |R|$,
  $\EMDuT(B,R)$ can be computed in time $\Oh(m n (\log n + \log^2 m))$.
\end{theorem}

Note that for $m = \Omega(n)$, for the asymmetric variant we obtain
near-quadratic time $\tOh(n^2)$, while for the symmetric variant we
obtain near-linear time $\tOh(n)$. We fully explain this gap, by
proving a matching conditional lower bound showing that no algorithm
solves the asymmetric variant in strongly subquadratic time
$\Oh(n^{2-\delta})$ for any $\delta > 0$, for $m = \Omega(n)$. 
In fact, we present a stronger lower bound that even rules out fast approximation algorithms, not only fast exact algorithms.  Our lower bound assumes the Orthogonal Vectors Hypothesis (OVH), a widely-accepted conjecture from fine-grained complexity theory; for a definition see Section~\ref{sec:1Dlowerbound}.

\begin{theorem}[1D Lower Bound]\label{thm:ov-lb}
  Assuming OVH, for any constant $\delta > 0$ there is no algorithm that,
  given $\eps \in (0,1)$ and sets $B,R \subseteq \R$ of size $n = |R| \ge |B| = \Omega(n)$, computes a $(1+\eps)$-approximation of $\EMDuT(B,R)$ in time $\Oh(n^{2-\delta} / \eps^{o(1)})$.
\end{theorem}

As a corollary, the same conditional lower bound holds for $\EMDuT_p$ over $\R^d$, for any $d \ge 1$ and $1 \le p \le \infty$, since subsets of $\R$ can be embedded into $\R^d$ for any dimension $d$ and any $L_p$ metric.

\medskip
Let us give a brief overview of these results.
In the symmetric setting, we establish that $f(\tau) \coloneq \EMD(B+\tau,R)$ is a
unimodal function in $\tau$, i.e., it is first monotone decreasing and then monotone increasing, and thus its minimum can be found easily.
In contrast, in the asymmetric setting the function $f(\tau)$ can have up to
$\Theta(n^2)$ disconnected global minima. Intuitively, our lower bound
shows that any algorithm needs to consider each one of these 
global near-minima, and therefore the running time must be quadratic in order to determine which near-minimum is the actual global minimum. To obtain our algorithm in the asymmetric setting, we use a sweep algorithm with an intricate event handling data structure.

\paragraph{EMDuT for \boldmath$L_1$ and \boldmath$L_\infty$ metric
  in higher dimensions.} We extend the work of
Eppstein~\etal~\cite{eppstein15improv_grid_map_layout_point_set_match}
for point sets in $\R^d$, leading to the following algorithms.

\begin{theorem}[Algorithms for $L_1$ and $L_\infty$ metric, Asymmetric]
  \label{thm:algo_higherdim}
  Given sets $B,R \subseteq \R^d$ of size $m = |B| \le n = |R|$,
  $\EMDuT_1(B,R)$ and $\EMDuT_\infty(B,R)$ can be computed in 
  $\Oh(m^dn^{d+2}\log^{d+2}n)$ time.
\end{theorem}

We explain that such a dependence on the dimension is unavoidable, by establishing
a more coarse-grained lower bound compared to our 1D results: We show that no
algorithm can solve the problem in time $n^{o(d)}$. In fact, we present a
stronger lower bound that even rules out fast approximation algorithms. Our
lower bound assumes the Exponential Time Hypothesis (ETH)~\cite{ImpagliazzoP01},
which is a well-established conjecture from fine-grained complexity theory.

\begin{theorem}[Lower Bound for $L_1$ and $L_\infty$ metric, Symmetric]
    Assuming ETH, there is no algorithm that,
    given $\eps \in (0,1)$ and sets $B,R \subseteq \R^d$ of size $n = |B| = |R|$, computes a $(1+\eps)$-approximation of $\EMDuT_1(B,R)$ in time $(\frac{n}{\eps})^{o(d)}$. The same holds for $\EMDuT_\infty(B,R)$.
\end{theorem}

Note that our lower bound pertains to the symmetric setting, while our algorithm
addresses the more general asymmetric setting. Hence, these results together
cover both the symmetric and the asymmetric setting.

\medskip
Let us give a brief overview of these results. For the algorithm, we establish
an arrangement of complexity $\Oh(m^d n^d)$ such that the optimal translation
$\tau$ is attained at one of the vertices within this arrangement. Our algorithm
is obtained by computing the \EMD at each vertex.
The lower bound is proven via a reduction from the $k$-Clique problem. In our
construction, each coordinate of the translation $\tau$ chooses one vertex from
a given $k$-Clique instance. We design gadgets that verify that every pair of
selected nodes indeed forms an edge.

\subsection{Open Problems}

\paragraph{EMDuT in 1D.}
Over $\R^1$, we leave open whether there are fast approximation algorithms: Can
a constant-factor approximation be computed in time $\Oh(n^{2-\delta})$ for some constant $\delta > 0$? Or even in time $\tOh(n)$? Can a $(1+\eps)$-approximation be computed in time $\tOh(n^{2-\delta} / \textup{poly}(\eps))$ for some constant $\delta > 0$ (independent of $n$ and $\eps$)? Or even in time $\tOh(n / \textup{poly}(\eps))$?

\paragraph{EMDuT for \boldmath$L_1$ and \boldmath$L_\infty$ metric in higher dimensions.}
For the $L_1$ and $L_\infty$ metric in dimension $d \ge 2$ we leave open to
determine the optimal constant $c>0$ such that the problem can be solved in time
$n^{c\cdot d + o(d)}$.

\paragraph{EMDuT for \boldmath$L_2$ metric in higher dimensions.}
The $L_2$ metric is the most natural measure in geometric settings, making
$\EMDuT_2$ a well motivated problem. The most pressing open problem is to determine the complexity of the $\EMDuT_2$ problem in any dimension $d \ge 2$.

We first observe that the $\EMDuT_2$ problem cannot be solved exactly (on the Real RAM model of computation). 
Indeed, on the Real RAM supporting only the usual arithmetic operations ($+,-,\cdot,/$), if the input numbers are rational then all output numbers are rational. If the machine further supports square roots (or other integral roots), if the input numbers are rational then the output numbers are algebraic. For the Geometric Median problem there are instances with rational input coordinates where the coordinates of the geometric median are not algebraic; this even holds in two dimensions~\cite{Bajaj88}. Therefore, Geometric Median cannot be solved exactly on the Real RAM. 
Finally, note that Geometric Median is a special case of $\EMDuT_2$, as for any
point set $R \subset \R^d$ of size $n$, if $B$ consists of $n$ copies of the
point $(0,\ldots,0)$, then $\EMDuT_2(B,R)$ is the (cost of the) Geometric Median
of $R$. 
Therefore, also $\EMDuT_2$ cannot be solved exactly on the Real RAM.

We therefore need to relax the goal and ask for an approximation algorithm.
Geometric Median has a very fast $(1+\eps)$-approximation algorithm running in
time $\Oh(nd\log^3(1/\eps))$~\cite{CohenLMPS16}, so the reduction from Geometric Median to $\EMDuT_2$ does not rule out very fast approximation algorithms for $\EMDuT_2$.

This is in stark contrast to what we know about the $\EMDuT_2$
problem, as almost all of our techniques in this paper completely fail
for this problem. We neither obtain an algorithm running in time
$n^{\Oh(d)}$, nor can we prove a lower bound ruling out time
$n^{o(d)}$. On the lower bound side, all we know is the lower bound from 1D, ruling out $(1+\eps)$-approximation algorithms running in time $\Oh(n^{2-\delta} / \eps^{o(1)})$ for any constant $\delta > 0$. On the algorithms side, one can observe that after fixing the matching from $B$ to $R$, the problem of finding the optimal translation $\tau$ for this matching is the Geometric Median problem and thus has a $(1+\eps)$-approximation algorithm running in time $\Oh(nd\log^3(1/\eps))$. By trying out all $n^{\Oh(n)}$ possible matchings, one can obtain a $(1+\eps)$-approximation algorithm for $\EMDuT_2$ running in time $n^{\Oh(n)} \log^3(1/\eps)$ for any constant $d$. We pose as an open problem to close this huge gap between the quadratic lower and exponential upper bound (for $(1+\eps)$-approximation algorithms with a $1/\eps^{o(1)}$ dependency on $\eps$ in the running time).


\section{Preliminaries}
We use $[n]$ to denote $\{1,\ldots,n\}$. All logarithms are base $2$. For every
$x \in \mathbb{R}$ we let $\cint{x} \in \mathbb{Z}$ be the unique integer such
that $x-\cint{x} \in (-1/2,1/2]$. Consider a set of blue points $B \subseteq \R^d$ and a set of red points $R \subseteq \R^d$. Fix an $L_p$ norm, for any $1 \le p \le \infty$. Denote by $\Phi$ the set of all injective functions $\phi \colon B \to R$, i.e., $\Phi$ is the set of all perfect matchings from $B$ to $R$. For any matching $\phi \in \Phi$ and any translation $\tau \in \R^d$ we define the cost
\[
  \D_{B,R,p}(\phi,\tau) = \sum_{b \in B} \|b+\tau-\phi(b)\|_p.
\]
We will ignore the subscript $p$ when it is clear from the context.
Note that we can express \EMD and \EMDuT in terms of this cost function as
\begin{align*}
    \EMD_p(B,R) = \min_{\phi \in \Phi} \D_{B,R,p}(\phi,(0,\ldots,0)) \,\text{
    and }\,
  \EMDuT_p(B,R) = \min_{\phi \in \Phi} \min_{\tau \in \R^d} \D_{B,R,p}(\phi,\tau).
\end{align*}

\section{Algorithm in One Dimension}
\label{sec:1D_Algorithm}

We first consider computing $\EMDuT_p(B,R)$ for two point sets $B,R$
in $\R^1$. For ease of presentation, assume that $R$ and $B$ are
indeed sets, and thus there are no duplicate points. We can handle the
case of duplicate points by symbolic perturbation. Observe, that the
distance between a pair of points $b,r$ in any $L_p$ metric is simply
$\|b-r\|_p = \|b-r\|_1 = |b-r|$. In Section~\ref{sub:1d_symmetric}, we
describe a very simple $\Oh(n\log n)$ time algorithm to compute
$\EMDuT_p(B,R)$ (as well as an optimal matching $\phi^*$ and
translation $\tau^*$ that realize this distance) when $B$ and $R$ both
contain exactly $n$ points. In Section~\ref{sub:Asymmetric_1d}, we
consider the much more challenging case where $|B|=m$ and $|R|=n$
differ. For this case we develop an $\Oh(nm(\log n + \log^2 m))$ time
algorithm to compute $\EMDuT_p(B,R)$.

A matching $\phi$ is said to be \emph{monotonically increasing} if and
only if for every pair of blue points $b' < b$ we also have
$\phi(b') < \phi(b)$. We show the
following crucial property.

\begin{restatable}{lemma}{monotonicallyIncreasing}
	\label{lem:1d_order}
	For any $B,R \subset \R$ there is an optimal matching $\phi$ that is
    monotonically increasing.
\end{restatable}

\begin{proof}
  We say that $(b,b') \in B \times B$ forms a
  \emph{crossing} in a matching $\phi$ if $b > b'$ and $\phi(b) < \phi(b')$.
  Let $\phi$ be an optimal matching of $B,R \subseteq \R$ with the minimal
  number of crossings.  If $\phi$ does not have any crossing, it is
  monotonically increasing. Hence, for the sake of contradiction assume that
  $(b,b')$ is a crossing in $\phi$.  Let $r = \phi(b')$ and $r' = \phi(b)$ and
  consider a matching $\phi'$ that has $\phi'(b) = r$ and
  $\phi'(b') = r'$, and $\phi'(x) = \phi(x)$ for every $x \in B \setminus
  \{b,b'\}$. 

  We will show that $\D_{B,R}(\phi,0) \ge \D_{B,R}(\phi',0)$. Combined
  with the fact that $\phi'$ has less crossings than $\phi$, this yields a
  contradiction to the choice of $\phi$ as the optimal matching with the minimal
  number of crossings.

  Note that $\D_{B,R}(\phi,0) \ge \D_{B,R}(\phi',0)$ is equivalent to
  \begin{align}\label{eq:1d_order}
    |b - r'| + |b'-r| \ge |b-r| + |b'-r'|.
  \end{align}
  Since $b' < b$ and $r' < r$, inequality~(\ref{eq:1d_order}) follows from the fact below (by setting $x = b'-r$, $\alpha = b-b'$, and $\beta = r-r'$).
\end{proof}
\begin{fact}\label{fact:abs-ineq}
    For every $x \in \R$ and $\alpha,\beta >0$ it holds that
    $|x| + |x+\alpha+\beta| \ge |x+\alpha| + |x+\beta|$.
\end{fact}

\subsection{Symmetric Case}
\label{sub:1d_symmetric}

In the symmetric case ($|R| = |B|$), Lemma~\ref{lem:1d_order} uniquely
defines an optimal matching. Let $B=\{b_1,\ldots,b_n\}$ and
$R=\{r_1,\ldots,r_n\}$ be the points in increasing order. Now, the optimal
translation $\tau^*$ is the value for $\tau$ that minimizes
$\D_{B,R}(\phi,\tau) = \sum_{i=1}^n |b_i - r_i + \tau|$.  Thus, it
corresponds to the median of $b_1 - r_1, \ldots, b_n - r_n$, which we
can compute in $\Oh(n\log n)$ time.

\begin{theorem}
	\label{thm:1d_algo_sym}
	We can compute $\EMDuT(R,B)$ in 1D in $\Oh(n \log n)$ time when $|R| = |B|$.
\end{theorem}

\subsection{Asymmetric Case} 
\label{sub:Asymmetric_1d}

In this section, we present an $\Oh(mn(\log n + \log^2 m))$ time
algorithm to compute $\EMDuT(B,R)$, for the case that $m \leq
n$. Consider the cost
$f(\tau) = \min_{\phi \in \Phi} \D_{B,R}(\phi,\tau)$ as a function of
$\tau$. The minimum of this function is $\EMDuT(B,R)$. The main idea
is then to sweep over the domain of $f$, increasing $\tau$ from
$-\infty$ to $\infty$, while maintaining (a representation of) $f$ and
a matching $\phi$ that realizes cost $f(\tau) =
\D_{B,R}(\phi,\tau)$. We also maintain the best translation
$\tau^* \leq \tau$ (i.e. with minimal cost) among the translations
considered so far (and if there are multiple such translations, the
smallest one), so at the end of our sweep, $\tau^*$ is thus an optimal
translation.

\paragraph{Properties of $f$.} By Lemma~\ref{lem:1d_order}, for any
$\tau$, there exists an optimal monotonically increasing matching
between $B+\tau$ and $R$. So, we restrict our attention to such
monotonically increasing matchings. Observe that any such matching
$\phi$ corresponds to a partition of $B$ into \emph{runs},
i.e. maximal subsequences of consecutive points, $B_1,\ldots,B_z$, so
that the points $b_{t-k},\ldots,b_t$ in a run $B_i$ are matched to
consecutive red points $r_{u-k},\ldots,r_u$, for some
$r_u=\phi(b_t)$. Moreover, for any such a matching $\phi$, the
function $\D_{B,R}(\phi,\tau)$ is piecewise linear in $\tau$, and each
breakpoint is a translation $\tau$ for which there is a pair
$(b,r) \in B\times R$ with $b+\tau = r$. It then follows that
$f(\tau)$ is also piecewise linear in $\tau$. Furthermore, the
breakpoints of $f$ are of two types. A type (i) breakpoint is a
translation such that there is a pair $(b,r) \in B\times R$ with
$b+\tau = r$, and a type (ii) breakpoint if there are two different
matchings $\phi,\phi'$ that both realize the same minimum cost
$\D_{B,R}(\phi,\tau) = \D_{B,R}(\phi',\tau)$. We show the
following key lemma, which lets us characterize the breakpoints of
type (ii) more precisely.

\begin{figure}[tb]
  \centering
  \includegraphics[width=0.6\textwidth]{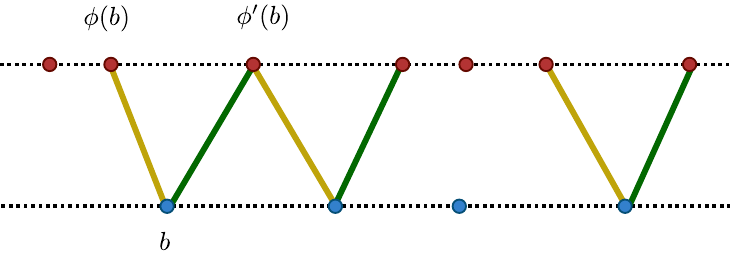}
  \caption{Schematic representation of the graph
    $G = \phi \oplus \phi'$ used in the proof
    of~\cref{lem:move_forward}. Each edge exists if and only if
    exactly one edge from either $\phi$ or $\phi'$ is present. Green
    edges arise from the matching $\phi'$, while yellow edges arise
    from the matching $\phi$.}
  \label{fig:xor-graph}
\end{figure}
\begin{restatable}{lemma}{moveForward}
  \label{lem:move_forward}
  Let $\phi$ be an optimal monotone matching of $\EMD_p(B + \tau,R)$, and let $\phi'$
  be an optimal monotone matching of $\EMD_p(B+\tau',R)$ for some $\tau' > \tau$.
  Then, $\phi'(b) \ge \phi(b)$ for all $b \in B$.
\end{restatable}

\begin{proof}
  We assume that all the points are distinct, as otherwise we can
  perturb them infinitesimally to resolve ties. Consider a bipartite
  graph $G$ defined as follows: the vertices of $G$ are $B \uplus R$,
  and we add an edge $(b,r) \in B \times R$ to $G$ if exactly one of
  the following conditions holds: (i) $\phi(b) = r$ and
  $\phi'(b) \neq r$, or (ii) $\phi'(b) = r$ and $\phi(b) \neq r$. The
  graph $G$ can be thought of as the exclusive-or of the matchings
  $\phi$ and $\phi'$, see \cref{fig:xor-graph}. We will now
  demonstrate that the connected components of this graph are
  paths. Then, considering that the matchings are monotone, it follows
  that the edges of these paths are non-crossing. This implies that
  consecutive red vertices on these paths are monotone. The lemma statement then
  easily follows.

  Let $C$ be any connected component of $G$ that consists of more than
  one vertex.
    \begin{claim}
        $C$ is a path.
    \end{claim}
    \begin{proof}
        Observe that the maximum degree of graph $G$ is $2$, so the connected
        components of $G$ consist of cycles and paths.
        Assume that $C$ is a cycle. In that case, however, there exists a
		pair of edges from either $\phi$ or $\phi'$ that intersect (see~\cref{fig:cycle}). This contradicts
        the assumption about the monotonicity of both $\phi$ and $\phi'$.
        Hence, $C$ is not a cycle, and the proof of the claim follows.
    \end{proof}
	\begin{figure}[tb]
	  \centering
	  \includegraphics[width=0.4\textwidth]{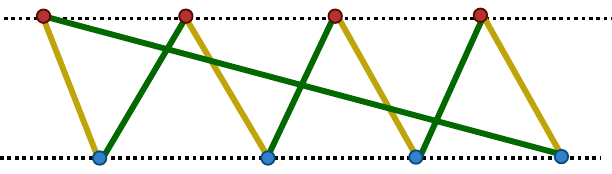}
	  \caption{Case when connected component of $G$ is a cycle.}
	  \label{fig:cycle}
	\end{figure}
    Now we know that $C$ is a path. Let $\{p_1,\ldots,p_\ell\} = V(C)$ be the
    consecutive vertices on the path~$C$. We have the freedom to
    select the order of endpoints; hence, without loss of generality,
    assume that $p_1$ is on the left of $p_\ell$, i.e., $p_1 < p_\ell$.
    \begin{claim}
        Vertices $p_1$ and $p_\ell$ are in $R$.
    \end{claim}
    \begin{proof}
        Every vertex $b \in B$ has degree $2$ in $G$ unless $b$ is an isolated vertex.
        However, both vertices $p_1$ and $p_\ell$ are endpoints of a path, which
        means that $p_1$ and $p_\ell$ have degree $1$ in $G$. This means that $p_1,p_\ell \notin B$.
    \end{proof}

    Since $G$ is a bipartite graph, this means that $\ell$ is odd.
    We assumed that $|V(C)| > 1$, so $\ell \ge 3$. Observe that if $i$ is odd, then $p_i \in R$, and if $i$ is even, then $p_i \in B$. Now,
    we show that red and blue vertices in $C$ are monotone:
    \begin{claim}\label{claim:increasing-seq}
        For every $i \in [\ell-2]$, it holds that $p_i < p_{i+2}$.
    \end{claim}
	\begin{proof}
		When $\ell = 3$, the claim holds because we have assumed $p_1 <
		p_3$. Hence, we can assume that $\ell \ge 5$ (since $\ell$ is odd).
		For the sake of contradiction, assume that $p_i > p_{i+2}$ for some $i$.
		Let $t_i \coloneqq p_{i+2} - p_i$ for $i \in [\ell-2]$. This means that
		there exists $i \in [\ell-2]$ with $t_i < 0$. Moreover, we have $p_\ell > p_1$,
		which means that $\sum_{\text{odd } i} t_i > 0$. Therefore, there also exists $j \in
		[\ell-2]$ with $t_j > 0$.

		In particular, there exists an index $k \in [\ell-2]$ such that $t_k \cdot
		t_{k+1} < 0$. This means that either (a) $p_k >
		p_{k+2}$ and $p_{k+1} < p_{k+3}$, or (b) $p_k < p_{k+2}$
		and $p_{k+1} > p_{k+3}$. In both of these cases, the intervals
        $(\min\{p_k,p_{k+1}\},\max\{p_k,p_{k+1}\})$ and
        $(\min\{p_{k+2},p_{k+3}\},\max\{p_{k+2},p_{k+3}\})$ intersect (see~\cref{fig:crossing}), which
		contradicts the assumption about the monotonicity of $\phi$
        and $\phi'$.
	\end{proof}
	\begin{figure}[tb]
	  \centering
	  \includegraphics[width=0.6\textwidth]{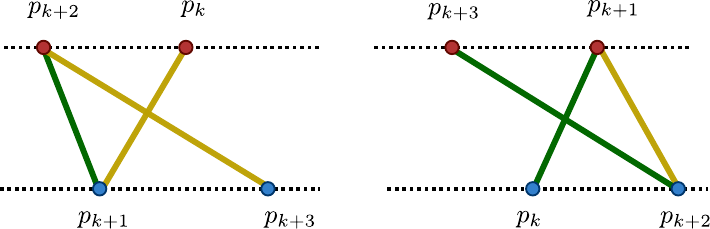}
	  \caption{Two cases of~\cref{claim:increasing-seq} in which a crossing occurs.}
	  \label{fig:crossing}
	\end{figure}

	Now, we continue the proof of~\cref{lem:move_forward}.
	For the sake of contradiction, assume that $\phi'(b) < \phi(b)$ for some $b
	\in B$. By~\cref{claim:increasing-seq}, there exists a connected component
    $C$ of $G$ such
    that $\phi'(b) <
    \phi(b)$ for every $b \in C \cap B$. Fix any such component $C$ and for every $b \in B$ let:
    \begin{align*}
        \psi(b) \coloneq \begin{cases}
            \phi'(b) & \text{ if } b \in C \cap B,\\
            \phi(b) & \text{ otherwise.}
        \end{cases}
        && \text{ and } && 
        \psi'(b) \coloneq \begin{cases}
            \phi(b) & \text{ if } b \in C \cap B,\\
            \phi'(b) & \text{ otherwise.}
        \end{cases}
    \end{align*}
    Notice that $\psi$ and $\psi'$ are both matchings of 
    $\EMD_p(B + \tau,R)$ and $\EMD_p(B + \tau',R)$ respectively. Now, recall that the cost of each matching is:
	\begin{align*}
        \D_{B,R}(\phi,\tau)   =& \sum_{b \in C\cap B} \abs{b+\tau-\phi(b)} + \sum_{b \in B\setminus C} \abs{b+\tau-\phi(b)},\\
        \D_{B,R}(\phi',\tau') =& \sum_{b \in C\cap B} \abs{b+\tau'-\phi'(b)} + \sum_{b \in B\setminus C} \abs{b+\tau'-\phi'(b)},\\
        \D_{B,R}(\psi,\tau) =& \sum_{b \in C\cap B} \abs{b+\tau-\phi'(b)} + \sum_{b \in B\setminus C} \abs{b+\tau-\phi(b)},  \text{ and }\\
        \D_{B,R}(\psi',\tau')  =& \sum_{b \in C\cap B} \abs{b+\tau'-\phi(b)} + \sum_{b \in B\setminus C} \abs{b+\tau'-\phi'(b)}.
	\end{align*}
	Since $\phi$ and $\phi'$ are optimal matchings for $\tau$ and $\tau'$
	respectively, we have:
    \begin{align}\label{eq:cost-matching}
		\D_{B,R}(\phi,\tau) \le \D_{B,R}(\psi,\tau)  && \text{ and } && \D_{B,R}(\phi',\tau') \le \D_{B,R}(\psi',\tau').
	\end{align}
        \newcommand{\RedR}{{R}}
        \newcommand{\BluB}{{B}}
        \newcommand{\typeone}  {\BluB\BluB\RedR\RedR\xspace}
        \newcommand{\typetwo}  {\RedR\RedR\BluB\BluB\xspace}
        \newcommand{\typethree}{\RedR\BluB\BluB\RedR\xspace}
        \newcommand{\typefour} {\BluB\RedR\BluB\RedR\xspace}
        \newcommand{\typefive} {\RedR\BluB\RedR\BluB\xspace}
        \newcommand{\typesix}  {\BluB\RedR\RedR\BluB\xspace}

        We say that $b \in B$ is a \emph{crossing} if $\phi'(b) < \phi(b)$.
        Notice that every $b \in C \cap B$ is a crossing.
        Next, we classify crossings into types based on the order of the points:
        \begin{itemize}
            \item\textbf{Type \typeone}: $b +\tau < b+\tau' \le \phi'(b) < \phi(b)$,
            \item\textbf{Type \typetwo}: $\phi'(b) < \phi(b) \le b+\tau < b+\tau'$,
            \item\textbf{Type \typethree}: $\phi'(b) \le b+\tau < b+\tau' \le \phi(b)$,
            \item\textbf{Type \typefour}: $b +\tau \le \phi'(b) < b+\tau' \le \phi(b)$,
            \item\textbf{Type \typefive}: $\phi'(b) \le b+\tau  < \phi(b) \le b+\tau'$,
            \item\textbf{Type \typesix}: $b +\tau \le \phi'(b) < \phi(b) \le b+\tau'$.
        \end{itemize}

\begin{figure}[tb]
  \centering
  \includegraphics[width=0.4\textwidth]{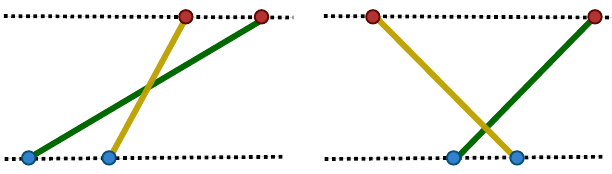}
  \caption{Illustration of crossing types. The left figure shows a \typeone crossing
  and the right figure shows a \typethree crossing.}
  \label{fig:suffix_functions}
\end{figure}
        Note that for any crossing $b$ of type \typeone or \typetwo it holds that:
        \begin{displaymath}
            |\phi(b) - b-\tau| + |\phi'(b) - b-\tau'| = |\psi(b) - b-\tau| + |\psi'(b) - b-\tau'|.
        \end{displaymath}
        Moreover, for any crossing $b$ of type \typethree, \typefour, \typefive or \typesix, denoting by $x_1 \le x_2 < x_3 \le x_4$ the numbers
        $b+\tau,b+\tau',\phi(b)$ and $\phi'(b)$ in sorted order, we have:
        \begin{align*}
            |\phi(b) - b-\tau| + |\phi'(b) - b-\tau'| & = x_3 + x_4 - x_1 - x_2,\\
            |\psi(b)-b-\tau| + |\psi'(b) - b - \tau'| & = x_2 + x_4 - x_1-x_3.\\
        \end{align*}
        Since $x_3 > x_2$, it follows that for any crossing $b$ of type \typethree, \typefour, \typefive or \typesix we have:
        \begin{displaymath}
            |\phi(b) - b-\tau| + |\phi'(b) - b-\tau'| > |\psi(b) - b - \tau| +
            |\psi'(b) - b-\tau'|.
        \end{displaymath}
        
        By summing up these inequalities over all $b \in C \cap B$, if at least one crossing in $C$ has type \typethree, \typefour, \typefive or \typesix we obtain
        \begin{align*}
            \D_{B,R}(\phi,\tau) + \D_{B,R}(\phi',\tau') > \D_{B,R}(\psi,\tau) + \D_{B,R}(\psi',\tau'),
        \end{align*}
        contradicting inequalities~(\ref{eq:cost-matching}).
        
        It remains to consider the case that all crossings in $C$ have type \typeone or \typetwo. 
        Note that if the leftmost crossing $b = \min (B \cap C)$ has type \typeone, then we can improve the matching $\phi$ by changing $\phi(b)$ to $\phi'(b)$; this contradicts the assumption that $\phi$ is an optimal monotone matching of $\EMD_p(B+\tau',R)$. Symmetrically, we obtain a contradiction if the rightmost crossing $b = \max (B \cap C)$ has type \typetwo. Finally, note that since $C$ is a connected component and both $\phi$ and $\phi'$ are
        monotone, if $C$ has a leftmost crossing of type \typetwo and a rightmost crossing of type \typeone, then in between these two there must be at least one crossing of type \typethree, \typefour, \typefive or \typesix. As this case was handled in the previous paragraph, we finished the proof.
\end{proof}

We can now come back to breakpoints and runs of $\phi$.
\begin{corollary}
  \label{cor:suffix_changes}
  A breakpoint $\tau$ of type (ii) corresponds to a pair of
  optimal monotonically increasing matchings $\phi, \phi'$ for which for all
  points $b \in B$ we have $\phi(b) \leq \phi'(b)$. Furthermore,
  consider a run $b_s,\ldots,b_{t}$ of $\phi$ and a point $b_i$ with
  $i \in \{s,\ldots,t\}$. If $\phi(b_i) < \phi'(b_i)$, then
  $\phi(b_j) < \phi'(b_j)$ for all $j \in \{i,\ldots,t\}$.
\end{corollary}
\begin{proof}
	Recall that $\tau$ is a breakpoint of type (ii) if there are two distinct matchings $\phi$ and $\phi'$ that both realize the optimal cost at $\tau$. Assume that $\phi$ is also optimal for a translation $\tau' < \tau$ or that $\phi'$ is also optimal for a translation $\tau'' > \tau$. Now by Lemma~\ref{lem:move_forward} we have $\phi(b) \leq \phi'(b)$ for every $b \in B$.
	Next, recall that $b_s, \dots, b_t$ form a run if $\phi$ matches them to consecutive red points. Since $\phi(b_{i+1})$ is matched to the red point right after $\phi(b_i)$, now $\phi'(b_i) > \phi(b_i)$ implies in particular that $\phi'(b_i) \geq \phi(b_{i+1})$. With monotonicity this gives us that  $\phi'(b_{i+1}) > \phi'(b_i) \geq \phi(b_{i+1})$. By induction, the same holds for the remaining $j \in \{i+2, \dots, t\}$.
\end{proof}

\begin{lemma}
  \label{lem:piecewise_linear_function_with_events}
  The function $f(\tau)$ is piecewise linear, and consists of
  $\Oh(nm)$ pieces.
\end{lemma}

\begin{proof}
  As argued above, $f$ is piecewise linear. What remains is to
  argue that there are $\Oh(nm)$ breakpoints. For every pair of points
  $(b_i,r_j) \in B \times R$ there is only one translation $\tau$ such
  that $b+\tau=r$, so clearly there are at most $\Oh(nm)$ breakpoints
  of type (i). At every breakpoint of type (ii), there is at least one
  blue point $b_i$ that was matched to $r_j$ and gets matched to some
  $r_k$ with $k > j$. This also happens at most once for every pair
  $b_i,r_j$. Hence, the number of breakpoints of type (ii) is also
  $\Oh(nm)$.
\end{proof}

In our sweep line algorithm, we will maintain a current optimal
matching $\phi$. At each breakpoint of type (i) we will have an event
to update the cost function of the matching. Furthermore, it follows
from Corollary~\ref{cor:suffix_changes} that when we sweep over a
breakpoint of type (ii), we can decompose the changes to the matching
using a series of \emph{atomic} events. In each such atomic event
there is some suffix $b_j,\ldots,b_t$ of a run $b_s,\ldots,b_t$ that $\phi$
currently matches to $r_{u-t+j},\dots,r_u$ that will become matched to
$r_{u-t+j+1},\ldots,r_{u+1}$. As we argued in the proof of
Lemma~\ref{lem:piecewise_linear_function_with_events}, the total
number of such events is only $\Oh(nm)$. Next, we express how we can
efficiently compute the next such atomic event, and handle it.

Consider a run $B_i=b_s,\ldots,b_t$ induced by $\phi$ at time $\tau$. Our aim
is to find the smallest $\tau' \geq \tau$ at which there is
an atomic type (ii) event involving a suffix $b_j,\ldots,b_t$ of
$B_i$. Hence, for a given suffix $b_j,\ldots,b_t$, we wish to
maintain when it starts being beneficial to match $b_j,\ldots,b_t$ to
$r_{u-t+j+1},\ldots,r_{u+1}$ rather than to $r_{u-t+j},\ldots,r_u$.

Let $\Delta'_j$ represent the change in cost
when we match $b=b_j$ to $r'=r_{v+1}$ rather than to $r=r_v$, ignoring that
$r_{v+1}$ may already be matched to some other blue point. We have that
\[
  \Delta'_j(\tau) = |b - r' + \tau| - |b - r + \tau| =
  \begin{cases}
    r' - r & \text{if } b + \tau \leq r, \\
    r + r' - 2b - 2\tau & \text{if } r < b + \tau < r', \\
    r - r' & \text{if } b + \tau \geq r'.
  \end{cases}
\]

\begin{figure}[tb]
  \centering
  \includegraphics[width=\textwidth]{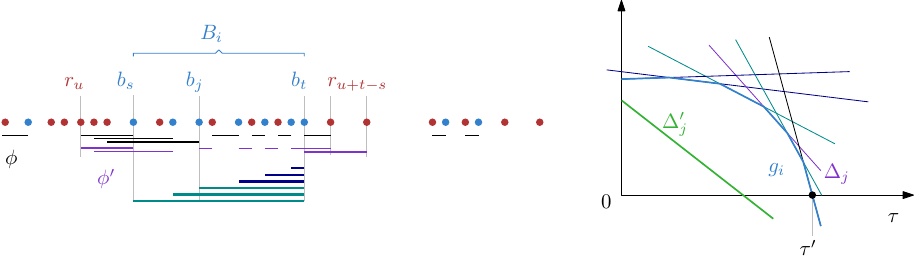}
  \caption{Each point $b_j$ in a run $B_i=b_s,\ldots,b_t$
    defines a (piecewise)-linear function $\Delta'_j$. Each suffix
    $b_j,\ldots,b_t$ then defines a linear function $\Delta_j$, expressing
    the cost of switching from matching $\phi$ to $\phi'$. The lower
    envelope $g_i$ of these functions then defines the first type (ii)
    event $\tau'$ of run $B_i$.}
  \label{fig:suffix_functions}
\end{figure}

Observe that this function is piecewise linear, and non-increasing. Moreover, the breakpoints coincide with type (i)
breakpoints of $f$ at which $b+\tau$ coincides with a red
point. Hence, in between any two consecutive events, we can consider
$\Delta'_j$ as a linear function. See Figure~\ref{fig:suffix_functions}
for an illustration.

We can then express the cost of changing the matching for the entire
suffix $b_j,\ldots,b_t$ as $\Delta_j(\tau) = \sum_{k=j}^t
\Delta'_k(\tau)$. This function is again decreasing, piecewise
linear, and has breakpoints that coincide with type (i) breakpoints of
$f$. When $\Delta_j(\tau)$ becomes non-positive it becomes beneficial
to match the suffix $b_j,\ldots,b_t$ to $r_{u-j+1},\ldots,r_{u+1}$. Hence,
the first such translation is given by a root of
$\Delta_j(\tau)$. Note that there is at most one such root since
$\Delta_j$ is decreasing.

It now follows that (if it exists) the root $\tau'$ of the function
$g_i(\tau) = \min_{j \in \{s,\ldots,t\}} \Delta_j(\tau)$ expresses the earliest
time that there is a suffix $b_j,\ldots,b_t$ for which it is beneficial to
update the matching. As before, this function is decreasing and
piecewise linear. Hence, we obtain:

\begin{lemma}
  \label{lem:run_event}
  Let $[\tau_1,\tau'] \ni \tau$ be a maximal interval on which
  $f(\tau)$ is linear, let $\tau'$ be a type (ii) breakpoint,
  and let $\phi$ be an optimal matching for $\tau$. Then there is a
  run $B_i$ induced by $\phi$, and $\tau'$ is a root of the
  function $g_i(\tau)$.
\end{lemma}

\paragraph{Representing the lower envelope $g_i$.} At any moment of
our sweep, we maintain a single piece of $g_i$. Hence, this piece is
the lower envelope of a set of linear functions
$\Delta_s, \ldots, \Delta_t$. We will maintain this lower envelope
using an adapted version of the data structure by Overmars and van
Leeuwen~\cite{overmars1981}. Ideally, we would maintain the lower
envelope of $\Delta_s, \ldots, \Delta_t$ directly. However,
reassigning a single blue point $b_j$ in the matching $\phi$, may
cause many functions $\Delta_k$ to change. So, we implicitly represent
each function $\Delta_j$ as a sum of $\Delta'_k$ functions.

\begin{restatable}{lemma}{lowerEnvelopeDataStructure}
  \label{lem:type_2_datastructure}
  Let $B_i$ be a run of size $k$. Using $\Oh(k\log k)$ space, we can represent
  the current piece of the lower envelope $g_i$ such that we can find
  the root of (this piece of) $g_i$ in $\Oh(\log k)$ time, and insert
  or remove any point in $B_i$ in $\Oh(\log^2 k)$ time.
\end{restatable}

\newcommand{\E}{\ensuremath{\mathcal{E}}\xspace}
\begin{proof}
  We will maintain this lower envelope using (a slightly adapted
  version of) the data structure by Overmars and van
  Leeuwen~\cite{overmars1981}. They present a data structure to store
  a lower envelope of $k$ lines while allowing queries such as line
  intersections in $\Oh(\log k)$ time, as well as insertions and
  deletions in $\Oh(\log^2 k)$ time. We could insert the lines
  (representing) $\Delta_s, \ldots, \Delta_t$ into this data structure
  so we can efficiently find the root of $g_i$. However, an update to
  $\phi$ causes an update to a function $\Delta'_k$, which in turn could cause changes
  in many functions $\Delta_j$, and could require significant
  changes to $g_i$. We therefore extend the data structure of Overmars
  and van Leeuwen to support such updates efficiently.

  \paragraph{The Overmars and van Leeuwen data structure.} The data
  structure by Overmars and van Leeuwen is a balanced binary tree,
  whose leaves store the lines $\Delta_s,\ldots,\Delta_t$ in the
  decreasing slope order. For each node $v$, let $\Delta^{v}$ denote the
  (ordered) set of lines stored in the leaves below $v$, and let $E^v$
  denote the ordered set of lines that appear in left-to-right order
  on the lower envelope $\E^v$ of $\Delta^v$. Overmars and van Leeuwen
  observe that (i) this set of lines $E^v$ defining the lower envelope
  $\E^v$ actually forms an ordered subset of the lines in $\Delta^v$,
  i.e. $E^v\subseteq \Delta^v$, and (ii) that $E^v$ consists of a
  prefix of $E^\ell$ concatenated with a suffix of $E^r$, where $\ell$
  and $r$ are the left and right child of $v$, respectively. See
  Figure~\ref{fig:overmars_van_leeuwen}.

  \begin{figure}[tb]
    \centering
    \includegraphics{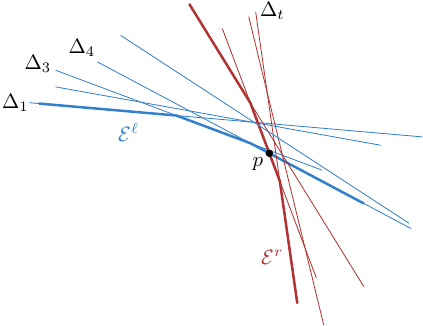}
    \caption{The set of lines $E^v$ contributing to the lower envelope
      $\E^v$ are ordered by decreasing slope. Furthermore, the lower
      envelopes $E^\ell$ and $E^r$ of the left and right child of $v$
      intersect in a single point $p$. Therefore, $E^v$ consists of a
      prefix of $E^\ell$ concatenated with a suffix of $E^r$.
    }
    \label{fig:overmars_van_leeuwen}
  \end{figure}

  The main idea is then to annotate each node $v$ with this ordered
  set $E^v$. Line intersection queries can then be answered in
  $\Oh(\log k)$ time by using the $E^v$ set of the root. To insert or
  delete a line $\Delta_j$, we follow a root-to-leaf path. 
  While walking back up the search path, we recompute $E^v$ from
  $E^\ell$ and $E^r$. In particular, by computing the intersection
  point of $\E^\ell$ and $\E^r$ and a constant number of split and
  concatenate operations on the involved ordered sets $E^v$. These
  operations can be implemented in $\Oh(\log k)$ time, thus leading to an
  $\Oh(\log^2 k)$ update time.

  Observe that the total size of all these $E^v$ sets is $\Oh(k\log
  k)$. Overmars and van Leeuwen reduce the space usage to $\Oh(k)$ by
  not explicitly storing the $E^v$ sets. Instead, they make sure that
  every node only stores the lines from $E^v$ that were not stored by
  the parent of $v$, and the relevant sets $E^v$ are reconstructed
  when performing updates. Furthermore, concatenating a prefix of
  $E^\ell$ with a suffix of $E^r$ to form $E^v$ typically destroys
  $E^\ell$ and $E^r$. They carefully describe how and which
  information to maintain to restore the sets appropriately.

  \paragraph{Our data structure.} We make two small, but important,
  changes to the above data structure.

  \begin{itemize}
  \item We observe that in our case the (line representing)
    $\Delta_j(\tau)$ never has greater slope than $\Delta_{j+1}$,
    because $\Delta_j(\tau) = \Delta'_j(\tau) + \Delta_{j+1}(\tau)$
    and $\Delta'_j$ has a non-positive slope. This means that, at any
    time, the lines $\Delta_s,\ldots,\Delta_t$ are already ordered by
    decreasing slope. Hence, we can use the indices of the functions
    to do the routing; i.e. each node $v$ will store the index $s^v$
    of the rightmost leaf in its left subtree rather than its slope.

  \item We extend the data structure so that given a linear function
    $f$ with non-positive slope and a value $j \in \{s,\ldots,t\}$, we can efficiently add $f$ to all
    linear functions $\Delta_j,\ldots,\Delta_t$.

    Observe that we are adding $f$ to the lines that already have
    the smallest slope, so this does not change the overall ordering,
    and adding $f$ to all functions in an ordered set
    $\Delta^v$ does not change the ordering of the (slopes of the)
    lines inside $\Delta^v$ either. Moreover, the combinatorial structure of its
    lower envelope $\E^v$ remains unchanged: a function $\Delta_a$ is
    the $h^\mathrm{th}$ function in $E^v$ if and only if $\Delta_a+f$
    is the $h^\mathrm{th}$ function in the lower envelope of
    $\{\Delta_c+f \mid \Delta_c \in E^v\}$.

    Hence, each node $v$ in our tree will store some additional linear
    function $f^v$ that we still have to add to all functions in the
    subtree rooted at $v$. As just argued, this allows us to represent
    the lower envelope $\E^v$ corresponding to node $v$ using $E^v$
    and $f^v$.

    We will use the same augmentation in the binary search trees that
    represent $E^v$ themselves; each node $\nu$ will store some linear
    function $f^\nu$ that should still be added to all functions
    stored in its subtree. We explicitly store these sets $E^v$ in
    persistent red black trees using path copying~\cite{driscoll89makin}, so that
    we can still have access to the original sets $E^\ell$ and $E^r$
    after ``combining'' them into $E^v$. This does not affect the
    running times: we can split, concatenate, and search using these
    structures in $\Oh(\log k)$ time, but it increases the total size of
    our structure to $\Oh(k\log k)$.

    Using the above representation, we can still compute the
    intersection of a (query) line with $\E^v$ in $\Oh(\log k)$ time:
    such a search follows some root-to-leaf path in the tree representing $E^v$,
    so at every step we can locally apply the offset $f^\nu$
    corresponding to the visited node $\nu$.

    Similarly, inserting or deleting a function still takes
    $\Oh(\log^2 k)$ time. When we visit node $v$, we ``push'' its offset
    towards its children. When moving back up the path, we compute the
    intersection point between two envelopes $\E^\ell$ and
    $\E^r$. This involves some simultaneous root-to-leaf path
    traversal of traversals of the trees representing $E^\ell$ and
    $E^r$, so we can again locally apply the offsets involved. Hence,
    this still takes $\Oh(\log k)$ time. We then split and concatenate
    $E^\ell$ and $E^r$ into $E^v$ in additional $\Oh(\log k)$
    time. (Note that this is where we use that every node $\nu$ in the
    tree representing $E^v$ can store its own offset $f^\nu$). Hence,
    the total time required is still $\Oh(\log^2 k)$.

    To add $f$ to all functions $\Delta_j,\ldots,\Delta_t$, we simply add
    $f$ to $f^r$, for all $\Oh(\log k)$ nodes $r$ hanging from the
    search path to $\Delta_j$ (in particular if the search path visits
    a node $v$ and its left child, we add $f$ to the right child $r$
    of $v$). Note that we can determine which child of node $v$ to
    visit by comparing $j$ against the index $s^v$. When we walk back
    up the path, we recompute $E^v$ from $E^\ell$ and $E^r$ as
    before. Hence, the total time required is $\Oh(\log^2 k)$.
\end{itemize}

\paragraph{Using the data structure.} We maintain the functions
$\Delta_s,\ldots,\Delta_t$ in the above data structure. The lower envelope
$\E^v$ represented by the root $v$ is then exactly the function
$g_i$. Hence, we can compute the root of $g_i$ by a line-intersection
query on $\E^v$. To insert a new point $b_j$, we insert the linear
function $\Delta_j$ into the structure, and increment the existing
functions $\Delta_{j+1},\ldots,\Delta_t$ by $\Delta'_j$. This takes
$\Oh(\log^2 k)$ time. Deleting a point is analogous.
\end{proof}

\paragraph{The main algorithm.} Our main algorithm sweeps the space
of all possible translations, while maintaining an optimal matching
$\phi$ for the current translation $\tau$, a representation of the
current piece of the function $f$ (i.e., the linear function $f'$ for
which $f(\tau)=f'(\tau)$), and the best translation $\tau^* \leq \tau$
found so far. To support the sweep, we also maintain a
Lemma~\ref{lem:type_2_datastructure} data structure for each run $B_i$
induced by $\phi$, and a global priority queue. The
Lemma~\ref{lem:type_2_datastructure} data structure allows us to
efficiently obtain the next type (ii) event of a run $B_i$. The global
priority queue stores all type (i) events, as well as the first type (ii)
event of each run.

\begin{figure}[tb]
  \centering
  \includegraphics{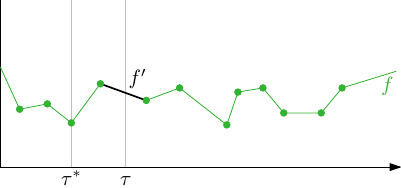}
  \caption{We sweep the domain of $f$, while maintaining a
    representation of the current piece $f'$ of $f$, and the best
    translation $\tau^* \leq \tau$ found so far. Breakpoints
    correspond to type (i) or type (ii) events.}
  \label{fig:sweep_functions}
\end{figure}

We initialize the priority queue by inserting all translations for
which a pair $(b,r) \in B \times R$ coincide as type (i) events. Let
$\tau_0$ be the first such event. For a translation $\tau < \tau_0$,
the matching $\phi$ that assigns $b_i$ to $r_i$ is optimal (by
Lemma~\ref{lem:1d_order}). Hence, we use $\phi$ as the initial
matching. We compute the corresponding function $f'$ expressing the
cost of $\phi$, construct the data structure of
Lemma~\ref{lem:type_2_datastructure} on the single run induced by
$\phi$, and query it for its first type (ii) event. We add this event
to the priority queue. All of this can be done in $\Oh(mn)$ time.

To handle a type (i) event involving point $b_j$, we remove it from
the data structure for its run and add it back in the same
place with its updated linear function $\Delta_j'(\tau)$. We query the
data structure to find the next type (ii) event of the run $B_i$
containing $b_j$, and update the event of $B_i$ in the global priority
queue if needed. Finally, if $b_j$ is aligned with $\phi(b_j)$ in the event,
we update $f'$ by adding the function $2(b_j+\tau-\phi(b_j))$
and evaluate it. Handling an event of type (i) takes $\Oh(\log n + \log^2 m)$ time, as it involves
a constant number of operations in the global priority queue, each
taking $\Oh(\log(nm)) = \Oh(\log n)$ time, and a constant number of
operations involving the Lemma~\ref{lem:type_2_datastructure} data
structures, each taking $\Oh(\log^2 m)$ time.

To handle a type (ii) event where the matching changes for points
$b_j, \ldots, b_t \in B_i$, we remove each point from the data
structure for $B_i$ and then add them to the run they are now
a part of (which can be either the existing run $B_{i+1}$ or a new run
in between $B_i$ and $B_{i+1}$). This takes
$\Oh(\log^2 m)$ time per point, but as argued in
Lemma~\ref{lem:piecewise_linear_function_with_events} each point can
only be involved in $\Oh(n)$ events of this type, so over all events,
this takes $\Oh(nm \log^2 m)$ time. We then
recompute the type (ii) events corresponding to the at most two
affected runs in $\Oh(\log m)$ time, and update them in the
global priority queue in $\Oh(\log n)$ time.
Here, we update $f'$ by adding the (linear) cost function $\Delta_i(\tau)$
associated with the event.

Thus, we handle a total of $\Oh(nm)$ events of type (i), each taking
$\Oh(\log n + \log^2 m)$ time, and $\Oh(nm)$ events of type (ii),
which take a total of $\Oh(nm(\log n + \log^2 m))$ time as well.

Once we have processed all events, the algorithm has found an optimal
translation $\tau^*$. We run the sweep once more from the start, and
stop at translation $\tau^*$, then report the current matching $\phi$
as an optimal matching. Together
with Theorem~\ref{thm:1d_algo_sym}, this thus establishes
Theorem~\ref{thm:1Dalgo}.


\section{Lower Bound in One Dimension}
\label{sec:1Dlowerbound}

\newcommand{\NearestRed}{\textup{NearestCell}}
\newcommand{\LHS}{\textup{LHS}}

In the Orthogonal Vectors problem (OV), we are given two sets of vectors $X,Y
\subseteq \{0,1\}^d$ with $n = |X| = |Y|$ and the task is to decide whether there exist $x \in
X$ and $y \in Y$ with $x \cdot y = 0$, where $x \cdot y = \sum_{i=1}^d x[i]
\cdot y[i]$. A naive algorithm solves this problem in time $\Oh(|X|^2 d)$.

\begin{hypothesis}[Orthogonal Vectors Hypothesis (OVH)~\cite{ov-seth,icm-survey}] \label{hyp:ovh}
    No algorithm solves the Orthogonal Vectors problem in time $\Oh(n^{2-\delta} d^c)$ for any constants $\delta,c > 0$.
\end{hypothesis}

In this section, we prove the following theorem.
\begin{theorem}\label{thm:1d-lb}
	Assuming OVH, for any constant $\delta > 0$ there is no algorithm that,
  given sets $B,R \subseteq \R$ of size $|R| \ge |B| = \Omega(n)$, computes $\EMDuT(B,R)$ in time $\Oh(n^{2-\delta})$. This even holds with the additional restriction $B,R \subseteq \{0,1,\ldots,\Oh(n^4)\}$.
\end{theorem}

Observe that this immediately implies~\cref{thm:ov-lb} because each
coordinate is bounded by a polynomial. Hence, from now on, we focus on
the proof of~\cref{thm:1d-lb}. 

We start by briefly sketching the reduction.
As a building block, we design vector gadgets that for two vectors 
$x \in X \subseteq \{0,1\}^d$ and $y \in Y \subseteq \{0,1\}^d$ are
sets of points $B(x)$ and $R(y)$ such that $\EMDuT(B(x),R(y)) = 0$ if $x$ and $y$ are
orthogonal and $\EMDuT(B(x),R(y)) \ge 1$ otherwise. This gadget is constructed
by encoding 
$x[i]$ and $y[i]$ coordinate-wise with the geometric patterns presented in~\cref{fig:vectors}.
Importantly, these gadgets share the same barycentre, therefore, when two
such patterns are translated far apart (i.e., their convex hulls do not
intersect), then their $\EMD$ does not depend on the coordinates of the vectors.
In~\cref{sub:gadgets} we formally construct these vector gadgets.

This allows us to construct a gadget such that for each
translation $\tau$ only a single pair of vector gadgets overlaps, and the
$\EMD$  between any other pair of vector gadgets does not depend on the
coordinates of the vectors.
In~\cref{sub:reduction} we show construct $B$ and $R$ as such collection of appropriately-spaced copies of
gadgets $B(x)$ and $R(y)$ for $x \in X$ and $y \in Y$ in such a way that there
exists a threshold $\Lambda$ such that
$\EMDuT(B, R) = \Lambda$ if there are two orthogonal vectors in $X \times Y$ and
$\EMDuT(B, R) \ge \Lambda + 1$ otherwise.
We formally prove these properties in~\cref{sub:equivalence,sub:properties}.

\subsection{Vector Gadgets}\label{sub:gadgets}

We construct two different types of gadgets depending on whether a vector belongs to set $X$ or $Y$ (see~\cref{fig:vectors} for illustration):
\begin{definition}[Red Vectors]
    For any vector $x \in \{0,1\}^d$, define a set of points $R(x)$ to consist of:
    \begin{itemize}
	\item $8d$ points at the coordinate $0$,
	\item $8d$ points at the coordinate $4d+1$, and
    \item for every $i \in \{1,\ldots,d\}$:
    \begin{itemize}
        \item if $x[i] = 0$, add points $\{4i-3,4i-2,4i-1,4i\}$,
        \item if $x[i] = 1$, add points $\{4i-2,4i-1\}$.
    \end{itemize}
    \end{itemize}
\end{definition}
\begin{definition}[Blue Vectors]
    For any vector $y \in \{0,1\}^d$, define a set of points $B(y)$ to consist of:
    \begin{itemize}
    \item one point at the coordinate $0$, 
    \item one point at the coordinate $4d+1$, and 
    \item for every $i \in \{1,\ldots,d\}$:
    \begin{itemize}
        \item if $y[i] = 0$, add points $\{4i-2,4i-1\}$,
        \item if $y[i] = 1$, add points $\{4i-3,4i\}$.
    \end{itemize}
    \end{itemize}
\end{definition}
\begin{figure}[tb]
  \centering
  \includegraphics[width=0.5\textwidth]{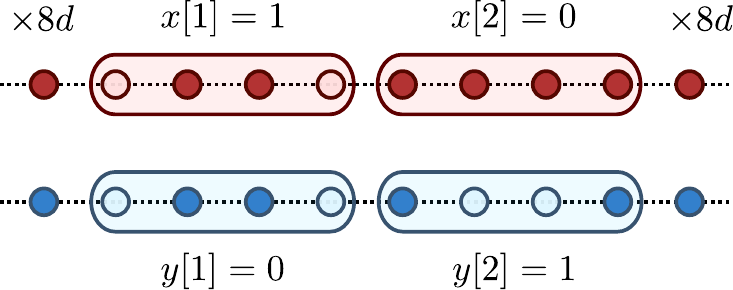}
  \caption{Gadgets for red and blue vectors in $d=2$. The top figure shows
  $R(x)$ for $x=(1,0)$, and the bottom figure shows $B(y)$ for $y=(0,1)$. Since $x$ and $y$ are orthogonal, each blue point corresponds to a red point with the same coordinate.}
  \label{fig:vectors}
\end{figure}

Next, we show that the above gadgets simulate the orthogonality. 

\begin{restatable}{lemma}{orthGadgets}\label{lem:orth-gadgets}
    Let $w = 4d+1$ be the width of the gadget.
    Let $x,y \in \{0,1\}^d$ be $d$-dimensional vectors. 
    \begin{enumerate}[noitemsep]
        \item\label{lem:orth-gadgets:it1} If $x$ and $y$ are orthogonal then
        $\EMD(B(y), R(x)) = 0$.
        \item\label{lem:orth-gadgets:it2} If $x$ and $y$ are not orthogonal then
            $\EMD(B(y)+\tau,R(x)) \ge \max\{1,|\tau|\}$ for all $\tau \in \mathbb{R}$.
    \end{enumerate}
    Moreover, if $|\tau| \ge w$, then we have $\EMD(B(y) + \tau,R(x)) = |\tau| \cdot
    c_1 - c_2$, where $c_1 = 2(d+1)$ and $c_2 = 4d^2+7d+1$.
\end{restatable}

\begin{proof}
    For the proof of Property~\ref{lem:orth-gadgets:it1}, assume that $x$ and
    $y$ are orthogonal. We construct a matching such that $\EMD(B(y),R(x)) = 0$.
    Fix an index $i \in \{1,\ldots,d\}$. If $y[i]=1$, then $x[i] = 0$; therefore, points at the coordinates $\{4i-3,4i\}$ exist in $R(x)$
    and we can match the points of $B(y)$ at these coordinates. If $y[i] = 0$, then
    note that the points at coordinates $\{4i-2,4i-1\}$ of $R(x)$ always exist. Hence, we can
    precisely match the points of $B(y)$ at these coordinates.

    Next, we prove Property~\ref{lem:orth-gadgets:it2}. Consider any $\tau \in
    \mathbb{R}$. At least one point of $B(y)$ (namely the leftmost or
    the rightmost point) must be matched by an edge of length $|\tau|$. This proves Property~\ref{lem:orth-gadgets:it2} if $|\tau| \ge 1$. So consider the case $|\tau| < 1$. If $x$ and $y$
    are not orthogonal, there exists an $i \in \{1,\ldots,d\}$ such that
    $x[i] = y[i] = 1$. By construction, $B(y)$ contains points at the
    coordinates $4i$ and $4i-3$, but $R(x)$ does not contain such points. Thus, the point $4i$ in $B(y)$ is matched to a point in distance at least $1 - |\tau|$. In addition to the cost of $|\tau|$ incurred by the leftmost or rightmost point of $B(y)$, we obtain
    $\EMD(B(y),R(x)) \ge 1$.

    Finally, we focus on the proof of the last property. Consider any $\tau \in
    \mathbb{R}$ with $|\tau|
    \ge w$. Observe that $B(y)+\tau$ contains $2(d+1)$ points in total. The closest
    point in $R(x)$ to each of them is at the coordinate $0$ or $w$. Because there
    are $8d$ such points, there are enough points to match each point in $B(y)$ to
    the closest point in $R(x)$. Hence, $\EMD(B(y)+\tau,R(x))$ matches each point of
    $B(y)$ to a point at the coordinate $0$ or $w$ of $R(x)$. If $\tau \ge w$, the cost of the matching is
\begin{align*}
    \EMD(B(y)+\tau, R(x)) & = 2\tau - w + \sum_{i=1}^{d} \tau - w + (4i -2 - y[i]) +
    \sum_{i=1}^{d} \tau - w + (4i-1+y[i]) \\
    & = 2 (d+1) \tau - 4d^2 - 7d-1.
\end{align*}
Consequently, we can take $c_1 = 2(d+1)$ and $c_2 = 4d^2+7d+1$. It remains to check that the above
choice is correct for $\tau \le -w$. In that case, the cost of the matching is
\begin{align*}
    \EMD(B(y)+\tau, R(x)) & = 2|\tau| - w + \sum_{i=1}^{d} |\tau| - (4i -2 - y[i]) + \sum_{i=1}^{d} |\tau| - (4i-1+y[i])\\
                          & = 2 (d+1) |\tau| - 4d^2-7d-1.
\end{align*}
Thus, the same choice of $c_1$ and $c_2$ works.
\end{proof}

\subsection{Reduction}\label{sub:reduction}

Now we use the vector gadgets from the previous section to reduce the
Orthogonal Vectors problem to \EMDuT. Specifically, given an OV instance $X, Y
\subseteq \{0,1\}^d$ such that $|X| = |Y| = n-1$, we construct sets $B,R \subseteq \R$ such that
from $\EMDuT(B,R)$ we can easily infer whether $X,Y$ contains an orthogonal pair
of vectors or not. Our reduction takes time $\Oh(nd)$ to construct the sets
$B,R$, in particular the constructed sets have size $\Oh(nd)$. Hence, if there
was an algorithm computing $\EMDuT(B,R)$ in time $\Oh(|R|^{2-\delta})$ for some constant $\delta > 0$, then our reduction would yield an algorithm for OV running in time $\Oh((nd)^{2-\delta})$, which contradicts OVH (Hypothesis~\ref{hyp:ovh}). That is, assuming OVH, $\EMDuT(B,R)$ cannot be computed in time $\Oh(|R|^{2-\delta})$ for any constant $\delta > 0$.

For the reduction, we can assume that $n$ is even, because otherwise we
can add a vector consisting exclusively of $1$s to both $X$ and $Y$.
We can also assume that $d \le n$, since otherwise the naive algorithm for OV already runs in time $\Oh(n^2d) = \Oh(n d^2)$.
Our reduction constructs the following point sets, for $\Delta \coloneqq 1000dn$:
\begin{itemize}
	\item \textbf{Red Points}: For the $i^{\mathrm{th}}$ vector $x_i \in X$, we create five red gadgets $R(x_i)^{(1)},\ldots,R(x_i)^{(5)}$. 
		For each $k \in [5]$, we translate $R(x_i)^{(k)}$ by
        $(i+kn)\cdot(n-1)\Delta$ and call it the $(i+kn)^\mathrm{th}$ \emph{red cell}.
    \item \textbf{Blue Points}: For the $j^{\mathrm{th}}$ vector $y_j \in Y$, we create a blue gadget $B(y_j)$ and translate it by $j \cdot n \Delta$. This set of points is called the $j^{\mathrm{th}}$ \emph{blue cell}.
\end{itemize}

We create five copies of red points for a technical reason that will become
clear later (just three copies would be enough, but then we would need to argue about
two types of optimal translations in the analysis).  We denote the set of all
red points by $\Red$, and the set of all blue points by $\Blue$. This concludes
the construction. Observe that $B,R$ can be constructed in time $\Oh(nd)$, as
claimed, and that their coordinates are in $\{0,\ldots,\Oh(dn^3)\} \subseteq
\{0,\ldots,\Oh(n^4)\}$.  Let $c_1$ and $c_2$ be the constants
(that depend on $d$) from~\cref{lem:orth-gadgets}. Let 
\begin{displaymath}
    \Lambda \coloneqq c_1 \Delta \cdot n(n-2)/4 + c_2\cdot (n-2).
\end{displaymath}
It remains to prove that the sets
$X, Y$ contain orthogonal vectors if and only if $\EMDuT(\Blue,\Red) \le
\Lambda$ (and thus from the value $\EMDuT(\Blue,\Red)$ we can easily infer
whether $X,Y$ contain orthogonal vectors).

\subsection{Properties of the Construction}
\label{sub:properties}

We now prove several useful properties of the construction
described above.
\begin{proposition}\label{prop:tau}
    There exist $i_\tau,j_\tau \in \{0,\ldots,n-1\}$ and $\eps \in (-1/2,1/2]$
    such that for
    \begin{equation} \label{eq:tau}
        \tau \coloneqq \left((n-1)(2n+i_\tau) - n j_\tau + \eps \right)\cdot \Delta.
    \end{equation}
    we have $\EMD(\Blue+\tau,\Red) = \EMDuT(\Blue,\Red)$.
\end{proposition}
\begin{proof}
	Let $\lambda \coloneqq (n-1)n \Delta$ and let $\tau$ be the optimal translation.
	If $\tau < \lambda$, then we claim that $\EMD(B+\tau,R) \ge
	\EMD(B+\tau+\lambda,R)$. This holds because (i) the right endpoint of
	$B+\tau+\lambda$ is to the left of right endpoint of $R$, and (ii) because
	$R$ is periodic, with period $\lambda$, each edge of the matching
	$\EMD(B+\tau,R)$ can be shifted by $\lambda$ (or matched to a strictly
    closer point, if the shifted point is unavailable).

	An analogous argument ensures that when $\tau > 2\lambda$, then
    $\EMD(B+\tau,R) \ge \EMD(B+\tau-\lambda,R)$. Hence we can assume that $\tau \in
    [\lambda,2\lambda]$. Similar arguments show also show that for $\tau \in
        [\lambda,2\lambda]$ we have $\EMD(B+\tau,R) = \EMD(B+\tau+\lambda,R)$.

    Observe, that for $\tau \in [\lambda,2\lambda]$, there exist
    $k,j \in \{0,\ldots,n-1\}$ and $\eps \in (-1/2,1/2]$ such that:
    \begin{displaymath}
        \tau  = ((n-1) \cdot (n+k+1) - j + \eps) \Delta.
    \end{displaymath}
    Let $i_1 := k+j-n+1$ and $i_2 := k+j+1$ and observe that either $i_1$ or $i_2$
    is in $\{0,\ldots,n-1\}$ and is a valid choice of $i$.
\end{proof}

From now, we assume that our optimal translation is of the form~\eqref{eq:tau}
and the parameters $i_\tau,j_\tau$ and $\eps$ are known. The idea behind this form is to
have the property that the $i_\tau^{\mathrm{th}}$ red cell and $j_\tau^{\mathrm{th}}$ blue cell ``nearly-align''.
Before we explain this in more detail, our goal is to show that $|\eps|$ is
small. 

To this end, we will use the following equality:
\begin{proposition}\label{prop:tautology}
    For every even integer $n$, any $j_\tau \in \{0,\ldots,n-1\}$ and $\eps \in (-1/2,1/2]$ it holds that:
	\begin{displaymath}
        \sum_{k \in \{1,\ldots,n-1\}}
    \min\{\abs{k-j_\tau+\eps},\abs{k-j_\tau+(n-1)+\eps},\abs{k-j_\tau-(n-1)+\eps}\}
        = \abs{\eps} + n(n-2)/4.
	\end{displaymath}
\end{proposition}
\begin{proof}
    We will consider the case $j_\tau < n/2$ as otherwise the reasoning is
    analogous. Then $\abs{k-j_\tau} \le
    \abs{k-j_\tau+n-1}$ and we need to show:
    \begin{align}\label{apx:eq-toproof}
		\sum_{k=1}^{n-1} \min\{\abs{k-j_\tau+\eps},\abs{k-j_\tau-(n-1)+\eps}\} =
        \abs{\eps} + n(n-2)/4.
	\end{align}
    The sum in~\eqref{apx:eq-toproof} can be stratified into:
    \begin{align*}
        |\eps| + 
        \sum_{k=1}^{j_\tau-1} \abs{k-j_\tau+\eps} +
        \sum_{k=j_\tau+1}^{n/2+j_\tau-1} \abs{k-j_\tau+\eps} +
        \sum_{k=n/2+j_\tau}^{n-1} \abs{k-j_\tau-(n-1)+\eps}
        .
	\end{align*}
    Now, we can deduce the sign of each absolute value, so we obtain:
    \begin{align*}
        |\eps| + 
        \sum_{k=1}^{j_\tau-1} (j_\tau-k-\eps) +
        \sum_{k=j_\tau+1}^{n/2+j_\tau-1} (k-j_\tau+\eps) +
        \sum_{k=n/2+j_\tau}^{n-1} (n-1-k+j_\tau-\eps)
        .
	\end{align*}
    Observe, that the number of $\eps$ with positive and negative signs is
    equal, hence these cancel out, which yields:
    \begin{align*}
        |\eps| + 
        \sum_{k=1}^{j_\tau-1} (j_\tau-k) + \sum_{k=j_\tau+1}^{n/2+j_\tau-1}
        (k-j_\tau) +
        \sum_{k=n/2+j_\tau}^{n-1} (n-1-k+j_\tau)
        .
	\end{align*}
    Finally, we change the summation index and conclude:
    \begin{align*}
        |\eps| + 
        \sum_{\ell=1}^{j_\tau-1} \ell + \sum_{\ell=1}^{n/2-1} \ell +
        \sum_{\ell=j_\tau}^{n/2-1} \ell
        = |\eps| + 2\sum_{\ell=1}^{n/2-1}\ell = |\eps| + n(n-2)/4.
	\end{align*}
    Here, we used the formula 
$\sum_{i=1}^{m} i =m(m+1)/2$ for $m = n/2-1$.
\end{proof}

\begin{property}\label{prop:small-eps}
    If $|\eps| > 0.1$, then $\EMD(\Blue+\tau,\Red) > \Lambda$.
\end{property}
\begin{proof}
    Let us bound the length of $\EMD(\Blue+\tau,\Red)$. For the $k^{\mathrm{th}}$ blue cell,
    the distance to any closest red point is at least:
    \begin{displaymath}
        \min_{\ell \in \mathbb{Z}} \left\{\abs{k\cdot n \Delta + \tau-\ell\cdot (n-1) \Delta}-2w\right\},
    \end{displaymath}
    because the width of every cell is $w = 4d+1$. Note that the number of points
    in each blue cell is exactly $2(d+1)$, hence: 
    \begin{align*}
        \LHS \coloneqq \EMD(\Blue+\tau, \Red) & \ge \sum_{k=1}^{n-1} 2(d+1) \cdot \left( \min_{\ell \in \mathbb{Z}}
            \left\{\abs{kn \cdot \Delta + \tau-\ell(n-1)\cdot \Delta}
            \right\}-2w  
        \right).
    \end{align*}
    Next, we plug in the definition of $\tau := ( (n-1)(2n+i_\tau)-n j_\tau +
    \eps) \Delta$:
    \begin{align*}
        \LHS & \ge - 4(d+1) w+ 
        2 (d+1) \Delta \cdot \sum_{k=1}^{n-1}  \min_{\ell \in \mathbb{Z}}
            \left\{\abs{(n-1)(2n+i_\tau-j_\tau+k-\ell) + k - j_\tau + \eps} \right \}.
\intertext{Observe that each summand is minimized when $|2n+i_\tau-j_\tau+k-\ell| \le 1$, hence:}
    \LHS & \ge - 4(d+1) w + 
    2 (d+1) \Delta \cdot \sum_{k=1}^{n-1}  \min \left\{
\abs{k-j_\tau+\eps},\abs{k-j_\tau+n-1+\eps}, \abs{k-j_\tau-(n-1)+\eps}\right\}.
\intertext{Next, we use~\cref{prop:tautology}:}
\LHS & \ge -4(d+1) w + 
    2 (d+1)  \Delta \cdot \left( |\eps| + n(n-2)/4 \right). \
\intertext{Finally, we use $0.1 \le |\eps|\le 0.5$ and the fact that $\Lambda =
    c_1 \Delta n(n-2)/4 + c_2 (n-2)$, where $c_1 = 2(d+1)$:}
\LHS & \ge  -4(d+1) w + 2(d+1) \Delta \abs{\eps} + \Lambda - (n-2)c_2 >
    \Lambda,
    \end{align*}
    where we used $|\eps| \ge 0.1$ and $\Delta \gg 4dn$.
\end{proof}
Consider a cell $k \in \{1,\ldots,n-1\}$ of $\Blue$. We define
$\NearestRed(k)$ as:
\begin{displaymath}
    \NearestRed(k) = \begin{cases}
        2n+i_\tau-j_\tau+k & \text{if } |k-j_\tau| < n/2,\\
        2n+i_\tau-j_\tau+k+1 & \text {if } k - j_\tau \ge n/2,\\
        2n+i_\tau-j_\tau+k-1 & \text {if } k - j_\tau \le -n/2.\\
    \end{cases}
\end{displaymath}
Recall that $n$ is even, and hence all the cases are covered. We show that $\NearestRed(k)$ is the index of the red cell that is closest to
the $k^{\mathrm{th}}$ cell of $\Blue$ after translation $\tau$, meaning that the optimal matching matches all points in the $k^{\mathrm{th}}$ blue cell to points in the $\NearestRed(k)^{\mathrm{th}}$ red cell:
\begin{property}\label{prop:match-close}
    If $|\eps| < 0.1$ then each point in the $k^{\mathrm{th}}$ blue cell is matched to some point in the $\NearestRed(k)^{\mathrm{th}}$ red cell.
\end{property}
\begin{proof}
    Observe that for any $\tau$ that satisfies~\cref{prop:tau}, it holds that
    every cell of $\Blue+\tau$ lies between two consecutive cells of $\Red$.
    Note that the consecutive red cells are at a shorter distance than
    the consecutive blue cells. Hence, in between two
    consecutive red cells lies at most one blue cell.

    Next, observe that each red cell consists of $\ge 16d$ points. Each blue cell
    consists of just $2(d+1)$ points. Therefore, each point cell of $\Blue+\tau$ is
    matched to a point in either the next left or the next right red cell.  The distance
    between the $\ell^{\mathrm{th}}$ red cell and the $k^{\mathrm{th}}$ blue cell is:
    \begin{displaymath}
        \abs{kn \Delta + \tau - \ell (n-1) \Delta} .
    \end{displaymath}
    After plugging in the value of $\tau$ this equals:
    \begin{displaymath}
        \abs{(2n+i_\tau-j_\tau+k - \ell)\cdot(n-1) - j_\tau + k+\eps} \cdot \Delta.
    \end{displaymath}
    Observe that when $\abs{k-j_\tau} < n/2$ this is minimized for $\ell =
    2n+i_\tau-j_\tau+k$ because $|\eps| < 0.1$. When $\abs{k-j_\tau} \ge n/2$ we have two cases based on the sign of $(k-j_\tau)$ which matches
    the definition of the $\NearestRed(k)$ cell because $|\eps|<0.1$.
\end{proof}

\subsection{Equivalence}\label{sub:equivalence}
We prove the equivalence by showing two implications separately.
\begin{lemma}
    If every pair of vectors $x \in X, y \in Y$ is not orthogonal, then
    $\EMD(\Blue+\tau,\Red) > \Lambda$.
\end{lemma}
\begin{proof}
We assume that $\tau$ is of form defined in~\cref{prop:tau} as otherwise
$\EMD(\Blue + \tau,\Red)$ is large. Moreover, by~\cref{prop:small-eps} we
can assume that $|\eps| \le 0.1$.  The distance between the left endpoint
of the $k^{\mathrm{th}}$ cell of $\Blue$ and the leftmost endpoint of the nearest cell
of $\Red$ is given by:
\begin{displaymath}
        \abs{(n-1) \cdot \NearestRed(k) \cdot \Delta - n k \cdot \Delta-\tau}
        .
\end{displaymath}
Because each pair of vectors in the instance is not orthogonal, by~\cref{lem:orth-gadgets}, we know that the smallest distance between points in each cell is at least $1$. Hence, after plugging in the definition of $\NearestRed(k)$ and by~\cref{lem:orth-gadgets}, the contribution of points in the $k^{\mathrm{th}}$ cell of $\Blue$ to $\EMD(\Blue+\tau,\Red)$ is:
\begin{displaymath}
    \Cost(k) \ge \begin{cases}
        1 & \text{if } k=j_\tau, \\
        |j_\tau-k+\eps| \cdot c_1 \Delta + c_2 & \text{if } \abs{k-j_\tau} < n/2, \\
		(n-1-|j_\tau-k+\eps| )\cdot c_1 \Delta + c_2 & \text{if } \abs{k-j_\tau} \ge n/2.\\
    \end{cases}
\end{displaymath}
Observe that $\Cost(k)$ is at least $c_1 \Delta \cdot
\min\{\abs{j_\tau-k+\eps},\abs{k-j_\tau-(n-1)+\eps},\abs{j_\tau-k+(n-1)+\eps}\} + c_2$ when $k \neq j_\tau$. Hence, the total length of $\EMD(\Blue+\tau,\Red)$ is at least
$\sum_{k=1}^{n-1} \Cost(k)$ which is bounded by:
\begin{align*}
    1 + (n-2)c_2 + c_1 \Delta \Big(\sum_{k \in \{1,\ldots,n-1\}\setminus \{j_\tau\}}
    \min\{&\abs{j_\tau-k+\eps}, \abs{k-j_\tau-(n-1)+\eps},\\&\abs{n-1+j_\tau-k+\eps}\}\Big)
	.
\end{align*}
We use~\cref{prop:tautology} to conclude that:
\begin{displaymath}
    \sum_{k=1}^{n-1} \Cost(k) \ge 1 +  c_2 (n-2) + c_1 \Delta \cdot n(n-2)/4 = 1 +
    \Lambda.\qedhere
\end{displaymath}
\end{proof}

\begin{lemma}
    If $x_i \in X$ and $y_j \in Y$ are orthogonal, then $\EMDuT(\Blue,\Red) \le \Lambda$.
\end{lemma}
\begin{proof}
	Take $\tau^\ast = ((n-1)(2n+i) - n j) \cdot \Delta$. Observe that the $(2n+i)^{\mathrm{th}}$ cell
    of $\Red$ and the $j^{\mathrm{th}}$ cell of $\Blue$ align. Because $x_i$ and $y_j$ are
    orthogonal, \cref{lem:orth-gadgets} guarantees that $\EMD(R(x_i),B(y_j)) =
    0$. By~\cref{prop:match-close}, all remaining cells are matched to their closest
	cell. Therefore, by~\cref{lem:orth-gadgets} the contribution of edges with endpoints in the $k^{\mathrm{th}}$ cell of
    $\Blue$ to $\EMD(\Blue+\tau^\ast,\Red)$ is
    \begin{align*}
        \Cost(k) &= \abs{(n-1) \NearestRed(k)  \cdot \Delta - n k \cdot \Delta- \tau^\ast} \cdot
        c_1  + c_2,
    \intertext{for any $k \in \{1,\ldots,n-1\} \setminus \{j\}$. By plugging in the
	exact values of $\NearestRed(k)$, with similar arguments as in the previous proof, we conclude that:}
        \Cost(k) &= \begin{cases}
            0 & \text{if } k = j, \\
            \abs{j - k} \cdot c_1 \Delta + c_2 & \text{if } \abs{k - j} < n/2, \\
            (n-1 - \abs{j - k}) \cdot c_1 \Delta + c_2 & \text{if } \abs{k - j} \ge n/2.
        \end{cases}
    \intertext{Observe that equivalently, we can write $\Cost(k) = c_1 \Delta
        \cdot \min\{\abs{j-k},\abs{n-1-j+k},\abs{n-1+j-k}\} + c_2$ when $k \neq j$.  Therefore, the total length
	of the matching is:}
	\sum_{k=1}^{n-1} \Cost(k) & = \sum_{k\in \{1,\ldots,n-1\}\setminus\{j\}} c_2 + \Delta c_1 \cdot \min\{\abs{j-k},\abs{n-1-j+k},\abs{n-1+j-k}\} \\
		& = c_2 (n-2) + \Delta c_1 \cdot n(n-2)/4 = \Lambda.
    \end{align*}
	where in the last inequality we used~\cref{prop:tautology}.
\end{proof}


\section{Lower Bounds in Higher Dimension} \label{sec:Lowerbound_Symmetric}

In this section, we prove conditional lower bounds for approximating
$\EMDuT$ with the $L_1$ or $L_\infty$ norm. Our lower bounds assume the popular Exponential Time Hypothesis (ETH), which postulates that the 3-SAT problem on $N$ variables cannot be solved in time $2^{o(N)}$~\cite{ImpagliazzoP01}.

\begin{theorem}
  \label{thm:lowerbound_eth}
  Assuming ETH, there is no algorithm that, given $\eps \in (0,1)$ and $B,R \subseteq \R^d$ of size $|B|=|R|=n$, computes a $(1+\eps)$-approximation of $\EMDuT_1(B,R)$ (or $\EMDuT_\infty(B,R)$) and runs in time $(\frac{n}{\eps})^{o(d)}$.
\end{theorem}

We prove our lower bounds by a reduction from the $k$-Clique problem: Given a
graph $G = (V,E)$ with $N$ nodes, decide whether there exist distinct nodes
$v_1,\ldots,v_k \in V$ such that $(v_i,v_j) \in E$ for all $1 \le i < j \le k$.
Here, we always assume that $k$ is constant. A naive algorithm solves the $k$-Clique problem in time $\Oh(N^k)$. It is well known that this running time cannot be improved to $N^{o(k)}$ assuming ETH.

\begin{theorem}[\cite{ChenHKX06}] \label{thm:cliqueETH}
Assuming ETH, the $k$-Clique problem cannot be solved in time $N^{o(k)}$.
\end{theorem}

In our lower bounds we will use the following lemma that combines gadgets
$(B_1,R_1),\ldots,(B_k,R_k)$ into a single instance $(B,R)$ whose cost is essentially the total cost of all gadgets. To prove this lemma, we simply place the gadgets sufficiently far apart.

\begin{restatable}[Gadget Combination Lemma]{lemma}{gadgetCombinationLemma}\label{lem:gadget-combination}
Let $1 \le p \le \infty$.
Given sets $B_1,R_1,\ldots,B_k,R_k \subset \R^d$ of total size $n$ with $|B_i| \le |R_i|$ for all $i \in [k]$, in time $\Oh(nd)$ we can compute sets $B,R \subset \mathbb{R}^d$ of total size~$n$ such that
\begin{displaymath}
    \EMDuT_p(B,R) = \min_{\tau \in \R^d} \sum_{i=1}^k \EMD_p(B_i + \tau,R_i).
\end{displaymath}
\end{restatable}
\begin{proof} 
The intuition is as follows. For a sufficiently large number $U$ we construct the sets $ B \coloneqq \bigcup_{i=1}^k B_i + (U \cdot i, 0,\ldots,0)$ and $R \coloneqq \bigcup_{i=1}^k R_i + (U \cdot i, 0,\ldots,0)$, i.e., we place the gadgets sufficiently far apart. Then one can argue that any optimal matching must match points in $B_i$ to points in $R_i$, and thus the $\EMDuT$ cost splits over the gadgets as claimed.

Now we provide the proof details. Let $\cal{B}$ be the bounding box of $\bigcup_{i=1}^k B_i \cup R_i$, and let $\Delta$ be the sum of all side lengths of $\cal{B}$ (i.e., $\Delta$ is the $L_1$ diameter of $\cal{B}$). We set $U \coloneqq (2n+5) \Delta$. We construct the sets $B,R$ as
$$ B \coloneqq \bigcup_{i=1}^k B_i + (U \cdot i, 0,\ldots,0), \qquad R \coloneqq \bigcup_{i=1}^k R_i + (U \cdot i, 0,\ldots,0). $$
Note that we have
$$ \EMDuT_p(B,R) = \min_{\tau \in \R^d} \EMD_p(B+\tau,R) \le \min_{\tau \in \R^d} \sum_{i=1}^k \EMD_p(B_i+\tau,R_i), $$
where the inequality follows from restricting the matching $\phi\colon (B+\tau) \to R$ to map points in $B_i+\tau$ to points in $R_i$ for all $i$.

For the opposite direction, by considering $\tau \coloneqq (0,\ldots,0)$ and considering any matching that maps points in $B_i+\tau$ to points in $R_i$ for all $i$, we observe
$$ \EMDuT_p(B,R) \le |R| \cdot \Delta \le n \Delta. $$

Now consider an optimal translation $\tau^*$, i.e., $\tau^*$ realizes $\EMDuT_p(B,R) = \EMD_p(B+\tau^*,R)$.
We claim that $\|\tau^*\|_p \le (n+2)\Delta$. Indeed, suppose for the sake of contradiction that $\|\tau^*\|_p > (n+2)\Delta$. Then any point in $B_i + \tau^*$ has distance more than $\|\tau^*\|_p - 2\Delta$ to any point in $R_i$.
It follows that if $\tau^*_1 \le 0$, then any point in $B_1 + \tau^*$ has distance more than $\|\tau^*\|_p - 2\Delta$ to any point in $R$, and thus $\EMD_p(B+\tau^*,R) > \|\tau^*\|_p - 2\Delta \ge n \Delta$. This contradicts $\EMD_p(B+\tau^*,R) = \EMDuT_p(B,R) \le n \Delta$, as shown above.
Similarly, if $\tau^*_1 \ge 0$, then any point in $B_k + \tau^*$ has distance more than $\|\tau^*\|_p - 2\Delta$ to any point in $R$, and we again arrive at a contradiction. Hence, we have $\|\tau^*\|_p \le (n+2)\Delta$.

Now consider an optimal matching $\phi \colon (B+\tau^*) \to R$. If $\phi$ matches any point in $B_i+\tau^*$ to any point in $R_j$ for some $i \ne j$, then it incurs a cost of at least $U - 2\Delta - \|\tau^*\|_p \ge U - (n+4)\Delta > n \Delta$, contradicting our upper bound $\EMDuT_p(B,R) \le n \Delta$. Therefore, $\phi$ matches points in $B_i+\tau^*$ to points in $R_i$ for all $i$, and we obtain
$$ \EMDuT_p(B,R) = \EMD_p(B+\tau^*,R) = \sum_{i=1}^k \EMD_p(B_i+\tau^*,R_i) \ge \min_{\tau \in \R^d} \sum_{i=1}^k \EMD_p(B_i+\tau,R_i). $$
Both directions together prove the lemma.
\end{proof}

For the readers' convenience, in Section~\ref{sec:lbhighdim_L1asymm}
as a warmup we prove the lower bound for the $L_1$ norm in the
asymmetric setting, i.e., we allow $|B|$ to be smaller than $|R| =
n$. Then in Section~\ref{sec:lbhighdim_L1symm} we strengthen this
lower bound to hold even in the symmetric setting $|B| = |R| =
n$. Finally, in Section~\ref{sec:lbhighdim_Linftysymm} we prove the
lower bound for the $L_\infty$ norm in the symmetric setting.

\subsection{Lower Bound for \boldmath$L_1$ Asymmetric}
\label{sec:lbhighdim_L1asymm}

\begin{figure}[tb]
  \centering
  \label{fig:3dgrid}
  \begin{subfigure}[b]{0.4\textwidth}
      \includegraphics[width=\textwidth]{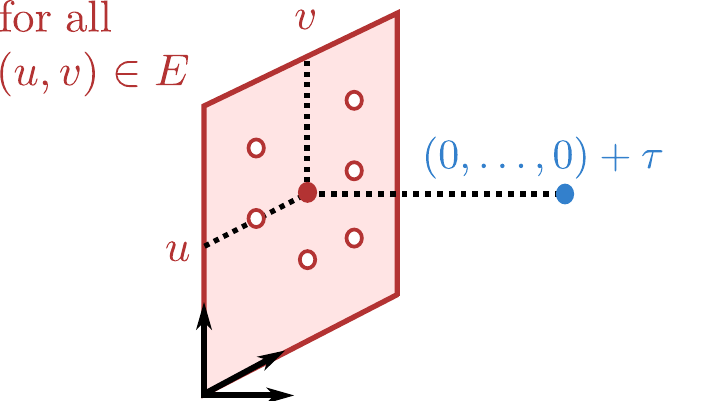}
  \end{subfigure}
  \hspace{2cm}
  \begin{subfigure}[b]{0.4\textwidth}
      \includegraphics[width=\textwidth]{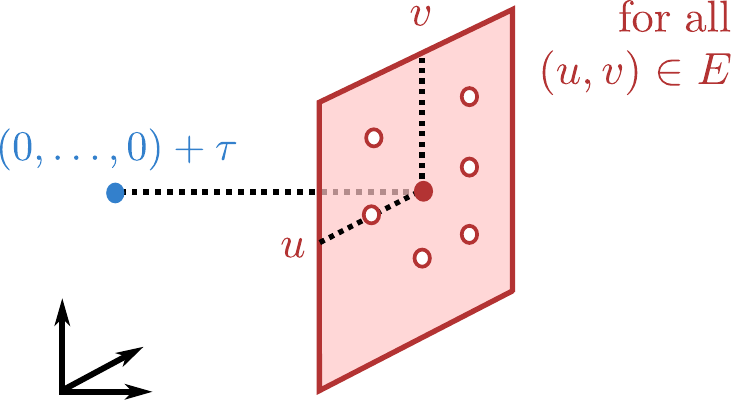}
  \end{subfigure}
  \caption{The left figure illustrates gadget $(B_{i,j},R_{i,j})$. The right figure
      illustrates gadget $(B_{i,j}',R_{i,j}')$.}
  \label{fig:3dgrid}
\end{figure}
In this section we prove Theorem~\ref{thm:lowerbound_eth} for the $L_1$ norm in the asymmetric setting, i.e., we relax the condition $|B|=|R|$ to $|B| \le |R|$.

We are given a $k$-Clique instance $G = ([N],E)$.
We set the dimension to $d \coloneqq k$.
In what follows by $p_{i,u,j,v,b} \in \R^d$ we denote the point with coordinates, for $\ell \in [d]$,
\begin{displaymath}
    (p_{i,u,j,v,b})_\ell = \begin{cases} u & \text{if } \ell = i, \\ v & \text{if } \ell = j, \\ b & \text{otherwise.} \end{cases}
\end{displaymath}
We construct the following $2 {k \choose 2}$ gadgets. For any $1 \le i < j \le k$ we construct
\begin{align*}
B_{i,j} &\coloneqq \{ (0,\ldots,0) \}, & R_{i,j} &\coloneqq \{ p_{i,u,j,v,0} \mid (u,v) \in E \}, \\
B'_{i,j} &\coloneqq \{ (0,\ldots,0) \}, & R'_{i,j} &\coloneqq \{ p_{i,u,j,v,N} \mid (u,v) \in E \}.
\end{align*}
The cost of these gadgets has the following properties.\footnote{Recall that
    $\cint{x}$ denotes the closest integer to $x$, while $[x]$ denotes
$\{1,\ldots,x\}$.}
\begin{lemma}
  Let $1 \le i < j \le k$.
  For any $\tau \in \R^d$ we have
  \begin{displaymath}
   \EMD_1(B_{i,j}+\tau, R_{i,j}) + \EMD_1(B'_{i,j}+\tau, R'_{i,j}) \ge (d-2)N, 
  \end{displaymath}
  and equality holds if $\tau \in [N]^d$ and $(\tau_i,\tau_j) \in E$.
  Moreover, for any $\tau \in \R^d$ with $(\cint{\tau_i},\cint{\tau_j}) \not\in E$ we have
  \begin{displaymath}
   \EMD_1(B_{i,j}+\tau, R_{i,j}) + \EMD_1(B'_{i,j}+\tau, R'_{i,j}) \ge (d-2)N + 1.
  \end{displaymath}
\end{lemma}
\begin{proof}
  Observe that
  \begin{align*}
  \EMD_1(B_{i,j}+\tau, R_{i,j}) = & \min_{(u,v) \in E} \|(0,\ldots,0)+\tau - p_{i,u,j,v,0}\|_1 \\&= \min_{(u,v) \in E} |\tau_i - u| + |\tau_j - v| + \sum_{\ell \ne i,j} |\tau_\ell|  \\
  &\ge \min_{(u,v) \in E} |\tau_i - u| + |\tau_j - v| + \sum_{\ell \ne i,j} \tau_\ell,
  \end{align*}
  where equality holds if $\tau \in [N]^d$.
  We similarly bound
  \begin{displaymath}
   \EMD_1(B'_{i,j}+\tau, R'_{i,j}) \ge \min_{(u,v) \in E} |\tau_i - u| + |\tau_j - v| + \sum_{\ell \ne i,j} N - \tau_\ell, 
  \end{displaymath}
  where equality holds if $\tau \in [N]^d$. Summing up and bounding the absolute values by 0, we obtain
  \begin{displaymath}
   \EMD_1(B_{i,j}+\tau, R_{i,j}) + \EMD_1(B'_{i,j}+\tau, R'_{i,j}) \ge (d-2) N.
  \end{displaymath}
  If $\tau \in [N]^d$ and $(\tau_i,\tau_j) \in E$, then we can pick $u,v$ with $|\tau_i - u| + |\tau_j - v| = 0$, and we obtain equality.

  Moreover, for any $\tau \in \R^d$ with $(\cint{\tau_i},\cint{\tau_j}) \not\in
  E$, note that since $(\tau_i,\tau_j)$ has $L_\infty$ distance at most $1/2$ to
  $(\cint{\tau_i},\cint{\tau_j})$, it has $L_\infty$ distance at least $1/2$ to any other grid point. In particular, $(\tau_i,\tau_j)$ has $L_\infty$ distance at least $1/2$ to any $(u,v) \in E$. Since $L_\infty$ distance lower bounds $L_1$ distance, we obtain $\min_{(u,v) \in E} |\tau_i - u| + |\tau_j - v| \ge 1/2$.
  This yields
  \begin{align*}
   \EMD_1(B_{i,j}+\tau, R_{i,j}) + \EMD_1(B'_{i,j}+\tau, R'_{i,j}) & \ge 2
   \min_{(u,v) \in E} \Big(|\tau_i - u| + |\tau_j - v|\Big) + (d-2) N\\ &\ge (d-2) N + 1. \qedhere
  \end{align*}
  
\end{proof}

We apply the Gadget Combination Lemma to the gadgets
$B_{i,j},R_{i,j},B'_{i,j},R'_{i,j}$ for $1 \le i < j \le d$. The value of the $\EMDuT_1$ of the resulting point sets $B,R$ is the sum of the costs of the gadgets. Hence, the above lemma implies the following.
If $G$ has a $k$-Clique $v_1,\ldots,v_k$, then $\tau \coloneqq (v_1,\ldots,v_k) \in [N]^d$ has a total cost of ${d \choose 2} \cdot (d-2) N \eqcolon \Lambda$.
On the other hand, if $G$ has no $k$-Clique, then for any $\tau \in \R^d$ there
exist $1 \le i < j \le k$ with $(\cint{\tau_i},\cint{\tau_j}) \not\in E$ (as
otherwise $(\cint{\tau_1},\ldots,\cint{\tau_k})$ would form a $k$-Clique). Thus, each pair of gadgets contributes cost at least $(d-2) N$, and  at least one pair of gadgets contributes cost at least $(d-2) N + 1$, so the total cost is at least ${d \choose 2} \cdot (d-2) N + 1 = \Lambda + 1$.

For any $\eps < 1 / \Lambda$, a $(1+\eps)$-approximation algorithm for
$\EMDuT_1$ could distinguish cost at most $\Lambda$ and cost at least $\Lambda +
1$, and thus would solve the $k$-Clique problem. Hence, if we had a $(1+\eps)$-approximation algorithm for $\EMDuT_1$ running in time $(n/\eps)^{o(d)}$, then by setting $\eps \coloneqq 0.9 / \Lambda$ and observing $n = \Oh(N^2)$, $1/\eps = \Oh(\Lambda) = \Oh(N)$, and $d=k$, we would obtain an algorithm for $k$-Clique running in time $(n/\eps)^{o(d)} = \Oh(N^3)^{o(k)} = N^{o(k)}$, which contradicts ETH by Theorem~\ref{thm:cliqueETH}.

\subsection{Lower Bound for \boldmath$L_1$ Symmetric}
\label{sec:lbhighdim_L1symm}

Now we strengthen the construction to work in the symmetric setting, where the number of blue and red points is equal. To this end, we add more blue points, and for technical reasons we also need to double the number of dimensions.

We are given a $k$-Clique instance $G=([N],E)$. We set the dimension to $d
\coloneqq 2k$. In what follows by $\bar p_{i,u,j,v,b} \in \R^d$ and $q_{i,j} \in
\R^d$ we denote the points with coordinates, for $\ell \in [d]$,
$$ (\bar p_{i,u,j,v,b})_\ell = \begin{cases} u & \text{if } \ell \in \{i,i+k\}, \\ v & \text{if } \ell \in \{j,j+k\}, \\ b & \text{otherwise,} \end{cases} \qquad (q_{i,j})_\ell = \begin{cases} N & \text{if } \ell \in \{i,j\}, \\ -N & \text{if } \ell \in \{i+k,j+k\}, \\ 0 & \text{otherwise.} \end{cases} $$
We construct the following $2 {k \choose 2}$ gadgets. For any $1 \le i < j \le
k$ we construct\footnote{Here, we treat $B_{i,j}$ as a multi-set, containing $|E|-1$ times the points $q_{i,j}$. This can be avoided by adding a tiny perturbation to each copy, which makes $B_{i,j}$ a set without significantly changing any distances.}
\begin{align*}
B_{i,j} &\coloneqq \{ (0,\ldots,0) \} \cup \{ |E|-1 \text{ copies of } q_{i,j} \}, & R_{i,j} &\coloneqq \{ \bar p_{i,u,j,v,0} \mid (u,v) \in E \}, \\
B'_{i,j} &\coloneqq \{ (0,\ldots,0) \} \cup \{ |E|-1 \text{ copies of } -q_{i,j} \}, & R'_{i,j} &\coloneqq \{ \bar p_{i,u,j,v,N} \mid (u,v) \in E \}.
\end{align*}
The cost of these gadgets has the following properties.
\begin{lemma}
  Let $1 \le i < j \le k$.
  For any $\tau \in \R^d$ we have
  $$ \EMD_1(B_{i,j}+\tau, R_{i,j}) + \EMD_1(B'_{i,j}+\tau, R'_{i,j}) \ge ((d+4)|E|-8)N, $$
  with equality if $\tau \in [N]^d$ and $(\tau_i,\tau_j) = (\tau_{i+k},\tau_{j+k}) \in E$.
  For any $\tau \in \R^d$ with $(\cint{\tau_i},\cint{\tau_j}) \not\in E$ we have
  $$ \EMD_1(B_{i,j}+\tau, R_{i,j}) + \EMD_1(B'_{i,j}+\tau, R'_{i,j}) \ge ((d+4)|E|-8)N + 1. $$
\end{lemma}
\begin{proof}
  Note that in any dimension $\ell \not\in\{i,j,i+k,j+k\}$ all points in $B_{i,j}$ and $B'_{i,j}$ have coordinate~0, and thus the contribution of these dimensions to the total cost does not depend on the matching. For $(B_{i,j},R_{i,j})$ such a dimension~$\ell$ contributes $|E| \cdot |\tau_\ell| \ge |E| \tau_\ell$ to the total cost, and for $(B'_{i,j},R'_{i,j})$ it contributes $|E| \cdot |N - \tau_\ell| \ge |E| (N - \tau_\ell)$. Summing up over both gadgets and all $\ell \not\in\{i,j,i+k,j+k\}$ yields total cost at least $(d-4) |E| N$, with equality if $\tau \in [N]^d$.

  In the remainder we can focus on the dimensions $i,j,i+k,j+k$. Projected onto these dimensions, the $L_1$ distance from $q_{i,j} + \tau$ to $\bar p_{i,u,j,v,0}$ is
  $$ |N + \tau_i - u| + |N + \tau_j - v| + |N - \tau_{i+k} + u| + |N - \tau_{j+k} + v|
  \ge 4N + \tau_i + \tau_j - \tau_{i+k} - \tau_{j+k}, $$
  where we used $|x| \ge x$. Furthermore, equality holds if $\tau \in [N]^d$.
  Analogously, one can show that projected onto the dimensions $i,j,i+k,j+k$ the $L_1$ distance from $-q_{i,j} + \tau$ to $\bar p_{i,u,j,v,N}$ is
  $$ \ge 4N - \tau_i - \tau_j + \tau_{i+k} + \tau_{j+k}, $$
  with equality if $\tau \in [N]^d$.
  Similarly, projected onto the dimensions $i,j,i+k,j+k$ the $L_1$ distance from $(0,\ldots,0) + \tau$ to $\bar p_{i,u,j,v,0}$ (or to $\bar p_{i,u,j,v,N}$) is
  $$ |\tau_i - u| + |\tau_j - v| + |\tau_{i+k} - u| + |\tau_{j+k} - v|. $$

  By summing up everything, and noting that the terms $\tau_i + \tau_j - \tau_{i+k} - \tau_{j+k}$ and $- \tau_i - \tau_j + \tau_{i+k} + \tau_{j+k}$ cancel, we obtain
  \begin{align*}
    &\EMD_1(B_{i,j}+\tau, R_{i,j}) + \EMD_1(B'_{i,j}+\tau, R'_{i,j})  \\
    &\ge 2\min_{(u,v) \in E} |\tau_i - u| + |\tau_j - v| + |\tau_{i+k} - u| + |\tau_{j+k} - v| + (|E|-1) \cdot 2 \cdot 4N + (d-4) |E| N,
  \end{align*}
  with equality if $\tau \in [N]^d$. Bounding the absolute values by 0 and simplifying, we obtain
  $$ \EMD_1(B_{i,j}+\tau, R_{i,j}) + \EMD_1(B'_{i,j}+\tau, R'_{i,j}) \ge ((d+4)|E|-8)N, $$
  with equality if $\tau \in [N]^d$ and $(\tau_i,\tau_j) = (\tau_{i+k},\tau_{j+k}) \in E$.

  Moreover, for any $\tau \in \R^d$ with $(\cint{\tau_i},\cint{\tau_j}) \not\in E$, the point $(\tau_i,\tau_j)$ has $L_\infty$ distance at least $1/2$ to any $(u,v) \in E$, and thus $\min_{(u,v) \in E} |\tau_i - u| + |\tau_j - v| \ge 1/2$. In this case we obtain
  $$ \EMD_1(B_{i,j}+\tau, R_{i,j}) + \EMD_1(B'_{i,j}+\tau, R'_{i,j}) \ge ((d+4)|E|-8)N + 2 \cdot \tfrac 12. \qedhere$$
\end{proof}

We apply the Gadget Combination Lemma to the gadgets
$B_{i,j},R_{i,j},B'_{i,j},R'_{i,j}$ for $1 \le i < j \le k$. The value of the $\EMDuT_1$ of the resulting point sets $B,R$ is the sum of the costs of the gadgets. Hence, the above lemma implies the following.
If $G$ has a $k$-Clique $v_1,\ldots,v_k$, then $\tau \coloneqq (v_1,\ldots,v_k,v_1,\ldots,v_k) \in [N]^d$ has a total cost of ${k \choose 2} \cdot ((d+4)|E|-8)N =: \Lambda$.
On the other hand, if $G$ has no $k$-Clique, then for any $\tau \in \R^d$ there
exist $1 \le i < j \le k$ with $(\cint{\tau_i},\cint{\tau_j}) \not\in E$ (as
otherwise $(\cint{\tau_1},\ldots,\cint{\tau_k})$ would form a $k$-Clique). Thus, each pair of gadgets contributes cost at least $((d+4)|E|-8)N$, and  at least one pair of gadgets contributes cost at least $((d+4)|E|-8)N + 1$, so the total cost is at least ${k \choose 2} \cdot ((d+4)|E|-8)N + 1 = \Lambda + 1$.

For any $\eps < 1 / \Lambda$, a $(1+\eps)$-approximation algorithm for
$\EMDuT_1$ could distinguish cost at most $\Lambda$ and cost at least $\Lambda +
1$, and thus would solve the $k$-Clique problem. Hence, if we had a $(1+\eps)$-approximation algorithm for $\EMDuT_1$ running in time $(n/\eps)^{o(d)}$, then by setting $\eps \coloneqq 0.9 / \Lambda$ and observing $n = \Oh(N^2), 1/\eps = \Oh(\Lambda) = \Oh(N^3)$, and $d=\Oh(k)$, we would obtain an algorithm for $k$-Clique running in time $(n/\eps)^{o(d)} = \Oh(N^5)^{o(k)} = N^{o(k)}$, which contradicts ETH by Theorem~\ref{thm:cliqueETH}.

\subsection{Lower Bound for \boldmath$L_\infty$ Symmetric}
\label{sec:lbhighdim_Linftysymm}

In this section we prove Theorem~\ref{thm:lowerbound_eth} for the $L_\infty$ norm, thus finishing the proof of this theorem.

We are given a $k$-Clique instance $G=([N],E)$. We set the dimension to $d \coloneqq 2k+1$. In what follows, by $\hat p_{i,u,j,v,b} \in \R^d$, $q_i \in \R^d$, and $b_{i,j} \in \R^d$ we denote the points with coordinates, for $\ell \in [d]$,
\begin{align*}
(\hat p_{i,u,j,v,b})_\ell &= \begin{cases} u & \text{if } \ell \in \{i,i+k\}, \\ v & \text{if } \ell \in \{j,j+k\}, \\ b & \text{otherwise,} \end{cases} \qquad (b_{i,j})_\ell = \begin{cases} 10N & \text{if } \ell \in \{i,j\}, \\ -10N & \text{if } \ell \in \{i+k,j+k\}, \\ 0 & \text{otherwise,} \end{cases}  \\
(q_i)_\ell &= \begin{cases} 10N & \text{if } \ell \in \{i,i+k\}, \\ 0 & \text{otherwise.} \end{cases}
\end{align*}
We construct the following $4(k-1)k + 2{k \choose 2}$ gadgets. For any $i \in [k], s \in [2(k-1)]$ we construct
\begin{align*}
B_{i,s} &\coloneqq \{ (0,\ldots,0) \}, & R_{i,s} &\coloneqq \{ q_i \}, \\
B'_{i,s} &\coloneqq \{ (0,\ldots,0) \}, & R'_{i,s} &\coloneqq \{ -q_i \}.
\end{align*}
Moreover, for any $1 \le i < j \le k$ we construct\begin{align*}
\hat B_{i,j} &\coloneqq \{ b_{i,j} \} \cup \{ |E|-1 \text{ copies of } (0,\ldots,0,10N) \}, & \hat R_{i,j} &\coloneqq \{ \hat p_{i,u,j,v,0} \mid (u,v) \in E \}, \\
\hat B'_{i,j} &\coloneqq \{ b_{i,j} \} \cup \{ |E|-1 \text{ copies of } (0,\ldots,0,-10N) \}, & \hat R'_{i,j} &\coloneqq \hat R_{i,j}.
\end{align*}
The following two lemmas analyze the properties of these gadgets.
\begin{lemma}
  Let $i \in [k]$ and $s \in [2(k-1)]$.
  For any $\tau \in \R^d$ we have
  $$ \EMD_\infty(B_{i,s}+\tau,R_{i,s}) + \EMD_\infty(B'_{i,s}+\tau,R'_{i,s}) \ge 20N + |\tau_i - \tau_{i+k}|, $$
  with equality if $\tau \in [N]^d$.
\end{lemma}
\begin{proof}
  We have
  $$ \EMD_\infty(B_{i,s}+\tau,R_{i,s}) = \|(0,\ldots,0) + \tau - q_i\|_\infty \ge \max\{ 10N - \tau_i, 10N - \tau_{i+k} \}, $$
  with equality if $\tau \in [N]^d$, as then the coordinates involving $10N$ dominate.
  Similarly, we have
  $$ \EMD_\infty(B'_{i,s}+\tau,R'_{i,s}) \ge \max\{ 10N + \tau_i, 10N + \tau_{i+k} \}, $$
  with equality if $\tau \in [N]^d$.
  We bound their sum by
  \begin{align*}
    &\EMD_\infty(B_{i,s}+\tau,R_{i,s}) + \EMD_\infty(B'_{i,s}+\tau,R'_{i,s})  \\
    &\ge \max\{(10N - \tau_i) + (10N + \tau_{i+k}), (10N - \tau_{i+k}) + (10N + \tau_i) \} = 20N + |\tau_i - \tau_{i+k}|.
  \end{align*}
  We observe that equality holds if $\tau \in [N]^d$, as then we have
  \begin{align*}
    &\EMD_\infty(B_{i,s}+\tau,R_{i,s}) + \EMD_\infty(B'_{i,s}+\tau,R'_{i,s})  \\
    &= \max\{(10N - \tau_i) + (10N + \tau_{i+k}), (10N - \tau_{i+k}) + (10N + \tau_i), \\
    &\qquad\quad\;\;\; (10N - \tau_i) + (10N + \tau_i), (10N - \tau_{i+k}) + (10N + \tau_{i+k}) \} = 20N + |\tau_i - \tau_{i+k}|.
    \qedhere
  \end{align*}
\end{proof}

\begin{lemma} \label{lem:intermediateLinftysymm}
  Let $1 \le i < j \le k$.
  For any $\tau \in \R^d$ we have
  $$ \EMD_\infty(\hat B_{i,j}+\tau, \hat R_{i,j}) + \EMD_\infty(\hat B'_{i,j}+\tau, \hat R'_{i,j}) \ge 20N |E| - 2|\tau_i - \tau_{i+k}| - 2|\tau_j - \tau_{j+k}|, $$
  with equality if $\tau \in [N]^d$ and $(\tau_i,\tau_j) = (\tau_{i+k},\tau_{j+k}) \in E$.
  For any $\tau \in \R^d$ with $(\cint{\tau_i},\cint{\tau_j}) \not\in E$ we have
  $$ \EMD_\infty(\hat B_{i,j}+\tau, \hat R_{i,j}) + \EMD_\infty(\hat B'_{i,j}+\tau, \hat R'_{i,j}) \ge 20N |E| - 2|\tau_i - \tau_{i+k}| - 2|\tau_j - \tau_{j+k}| + 1. $$
\end{lemma}
\begin{proof}
  Since the last coordinate of any point in $\hat R_{i,j}$ (or $\hat R'_{i,j}$) is 0, its distance to $(0,\ldots,0,10N)+\tau$ is at least $10N + \tau_d$, with equality if $\tau \in [N]^d$. Similarly, its distance to $(0,\ldots,0,-10N) + \tau$ is at least $10N - \tau_d$, with equality if $\tau \in [N]^d$. Therefore, all copies of $(0,\ldots,0,10N)$ and $(0,\ldots,0,-10N)$ in total contribute a cost of at least $2 (|E|-1) \cdot 10N$, with equality if $\tau \in [N]^d$, no matter what points in $\hat R_{i,j}$ (or $\hat R'_{i,j}$, resp.) they are assigned to.

  The remaining point $b_{i,j} + \tau$ has distance to $\hat R_{i,j}$ (and to $\hat R'_{i,j}$) of
  \begin{align} \label{eq:intermediateeq}
    \ge \min_{(u,v) \in E} \max\{ 10N + \tau_i - u, 10N + \tau_j - v, 10N - \tau_{i+k} + u, 10N - \tau_{j+k} + v\}.
  \end{align}
  Equality holds for $\tau \in [N]^d$, as then then coordinates involving $10N$ dominate. We further bound (\ref{eq:intermediateeq}) from below by
  \begin{align*}
  &= 10N + \min_{(u,v) \in E} \max\{ \tau_i - u, \tau_j - v, u - \tau_{i+k}, v - \tau_{j+k} \}  \\
  &\ge 10N + \min_{(u,v) \in E} \max\{ \tau_i - u, u - \tau_i, \tau_j - v, v - \tau_j\} - |\tau_i - \tau_{i+k}| - |\tau_j - \tau_{j+k}|  \\
  & = 10N + \min_{(u,v) \in E} \max\{ |\tau_i - u|, |\tau_j - v| \} - |\tau_i - \tau_{i+k}| - |\tau_j - \tau_{j+k}|.
  \end{align*}
  Equality holds if $\tau \in [N]^d$ and $(\tau_i, \tau_j) = (\tau_{i+k}, \tau_{j+k})$.

  Summing up these costs, in total we obtain
  \begin{align*}
    &\EMD_\infty(\hat B_{i,j}+\tau, \hat R_{i,j}) + \EMD_\infty(\hat B'_{i,j}+\tau, \hat R'_{i,j})  \\
    &\ge 20 N |E| + 2 \min_{(u,v) \in E} \Big( \max\{ |\tau_i - u|, |\tau_j - v| \} - |\tau_i - \tau_{i+k}| - |\tau_j - \tau_{j+k}| \Big),
  \end{align*}
  with equality if $\tau \in [N]^d$ and $(\tau_i, \tau_j) = (\tau_{i+k}, \tau_{j+k})$.
  By bounding the $\max$ term by 0, we obtain
  $$ \EMD_\infty(\hat B_{i,j}+\tau, \hat R_{i,j}) + \EMD_\infty(\hat B'_{i,j}+\tau, \hat R'_{i,j}) \ge 20N |E| - 2|\tau_i - \tau_{i+k}| - 2|\tau_j - \tau_{j+k}|. $$
  Equality holds if $\tau \in [N]^d$ and $(\tau_i, \tau_j) = (\tau_{i+k}, \tau_{j+k}) \in E$.

  Moreover, for any $\tau \in \R^d$ with $(\cint{\tau_i},\cint{\tau_j}) \not\in E$, the point $(\tau_i,\tau_j)$ has $L_\infty$ distance at least $1/2$ from any $(u,v) \in E$, and thus $\min_{(u,v) \in E} \max\{ |\tau_i - u|, |\tau_j - v| \} \ge 1/2$. In this case, we obtain cost at least
  $$ 20N |E| + 2\cdot \tfrac 12 - 2|\tau_i - \tau_{i+k}| - 2|\tau_j - \tau_{j+k}|. \qedhere$$
\end{proof}

We now apply the Gadget Combination Lemma to the gadgets
$B_{i,s},R_{i,s},B'_{i,s},R'_{i,s}$ for $i \in [k]$ and $s \in [2(k-1)]$ and $\hat
B_{i,j},\hat R_{i,j},\hat B'_{i,j},\hat R'_{i,j}$ for $1 \le i < j \le k$. The
value of $\EMDuT_\infty$ of the resulting point sets $B,R$ is the sum of the costs
of the gadgets. Hence, the cost can be read off from the above lemmas. Note that for any $i \in [k]$ the terms $|\tau_i - \tau_{i+k}|$ cancel, as there are $2(k-1)$ gadget pairs contributing $+|\tau_i - \tau_{i+k}|$ (one gadget pair for each $s \in [2(k-1)]$), and $k-1$ gadget pairs contributing $-2 |\tau_i - \tau_{i+k}|$ (one gadget pair for each $j \in [k], j \ne i$). The cost thus simplifies to
$$ \EMD_\infty(B+\tau,R) \ge 20N \cdot k \cdot 2(k-1) + 20N |E| \cdot {k \choose 2} \eqcolon \Lambda. $$
We obtain equality for any translation $\tau \in [N]^d$ with $(\tau_i,\tau_j) = (\tau_{i+k},\tau_{j+k}) \in E$ for all $1 \le i < j \le k$.
Therefore, if $G$ has a $k$-Clique $v_1,\ldots,v_k$, then $\tau \coloneqq (v_1,\ldots,v_k,v_1,\ldots,v_k,0) \in [N]^d$ has a total cost of $\Lambda$.
On the other hand, if $G$ has no $k$-Clique, then for any $\tau \in \R^d$ there
exist $1 \le i < j \le k$ with $(\cint{\tau_i},\cint{\tau_j}) \not\in E$ (as
otherwise $(\cint{\tau_1},\ldots,\cint{\tau_k})$ would form a $k$-Clique).
Then according to Lemma~\ref{lem:intermediateLinftysymm} at least one summand has an additional $+1$, so the total cost is at least $\Lambda + 1$.

For any $\eps < 1 / \Lambda$, a $(1+\eps)$-approximation algorithm for
$\EMDuT_\infty$ could distinguish cost at most $\Lambda$ and cost at least
$\Lambda + 1$, and thus would solve the $k$-Clique problem. Hence, if we had a $(1+\eps)$-approximation algorithm for $\EMDuT_\infty$ running in time $(n/\eps)^{o(d)}$, then by setting $\eps \coloneqq 0.9 / \Lambda$ and observing $n = \Oh(N^2), 1/\eps = \Oh(\Lambda) = \Oh(N^3)$, and $d=\Oh(k)$, we would obtain an algorithm for $k$-Clique running in time $(n/\eps)^{o(d)} = \Oh(N^5)^{o(k)} = N^{o(k)}$, which contradicts ETH by Theorem~\ref{thm:cliqueETH}.


\section{Algorithms in Higher Dimensions}
\label{sec:L1_and_L_infty_higherd}

Given two sets $R$ and $B$ of $n$ points in the plane, Eppstein
\etal~\cite{eppstein15improv_grid_map_layout_point_set_match} show how
to compute a translation $\tau^*$ minimizing $\EMDuT_1(B,R)$ with
respect to the $L_1$-distance in $\Oh(n^6\log^3 n)$ time. We observe
that their result can be generalized to point sets in arbitrary
dimension $d$, leading to an $\Oh(m^dn^{d+2}\log^{d+2}n)$ time
algorithm.

Furthermore, we show that our approach can be used to obtain an
$\Oh(m^dn^{d+2}\log^{d+2}n)$ time algorithm for finding a translation
that minimizes $\EMDuT_\infty(B,R)$, i.e. the Earth Mover's Distance
with respect to the $L_\infty$ distance. For point sets in $\R^2$,
this immediately follows by ``rotating the plane by $45^\circ$'' and
using the algorithm for $L_1$. For higher dimensions this trick is no
longer immediately applicable. However, we show that our algorithm can
also directly be applied to the $L_\infty$ distance, even for point
sets in $\R^d$ with $d > 2$.

\paragraph{Earth Mover's Distance without Translation.} We first
describe an algorithm to compute $\EMD_p(B,R)$ in $\mathbb{R}^d$.
Note that we assume to work in the Real RAM model, hence we need a
strongly-polynomial algorithm. Naively, one can achieve that in $\Oh(m^2n)$ time by computing the bipartite graph, and solving maximum weight matching in bipartite graph in strongly polynomial time by Edmonds and Karp~\cite{edmonds1972theoretical}. Here, however, we can use the fact that points are in $\mathbb{R}^d$. To the best of our knowledge, the best algorithm
in this setting is due to Vaidya~\cite{vaidya89geomet_helps_match}. However,
he only considers the case when both point sets are in $\R^2$ and have size
$n=m$ in $\R^2$. He shows that one can compute $\EMD_p(B,R)$ (with $p
\in \{1,\infty\}$) in $\Oh(n^2 \log^3
n)$ time in this setting. Furthermore, he states (without proof) that for point
sets in $\R^d$, that the running time increases by at most $\Oh(\log^d
n)$. Next, we briefly sketch the algorithm and fill in
the missing details for the higher-dimensional setting, to obtain the
following result:

\begin{restatable}{theorem}{fixedSetsAlgorithm}
  \label{thm:l1_fixed_sets}
  Given a set $B$ of $m$ points in $\R^d$, and a set of $n \geq m$ red
  points in $\R^d$, there is an $\Oh(n^2\log^{d+2}n)$ time algorithm
  to compute $\EMD_p(B,R)$, for $p \in \{1,\infty\}$.
\end{restatable}

\begin{proof}
  Vaidya's algorithm is a particular implementation of the Hungarian
  Method~\cite{kuhn1955hungarian}. Hence, to apply it we introduce
  $n-m$ additional ``dummy'' points in $B$. We define the distance
  from a dummy point to any other point in $R$ to be zero. We stress that the
  dummy points are only present in the graph representation, and are not
  physically in $\R^d$. The main
  algorithm proceeds in phases, in each of which the current matching
  grows by one new pair. Hence, the algorithm completes after $n$
  phases. In each phase, each point $q \in B \cup R$ is assigned a
  weight $w_q$ (the current value of the point in the dual
  LP-formulation of the problem). The algorithm then maintains a
  subset $R'$ of red points (initially $R'=R$), and a subset $B'$ of
  blue points (initially, $B'$ is the set of unmatched blue points),
  and repeatedly computes the bichromatic closest pair $(b^*,r^*)$
  with respect to the weighted distance function
  $d(b,r)=\|b-r\|_p - w_b - w_r$. This pair is either added to the
  matching (if $r^*$ is also unmatched), or the weights of $b^*$ and
  $r^*$ are updated, $r^*$ is removed from $R'$, and $b^*$ is added to
  $B'$. Since there are only $n$ points in one set, a phase consists of at
  most $\Oh(n)$ such steps. Hence, given a dynamic data structure
  storing $R'$ and $B'$ that
  \begin{itemize}[noitemsep]
  \item maintains the bichromatic (weighted) closest pair among $R'
    \cup B'$,
  \item can be built in $P(n)$ time (and thus uses at most $P(n)$ space),
  \item an insertion of a blue point into $B'$ in (amortized) $I(n)$ time, and
  \item a deletion of a red point from $R'$ in (amortized) $D(n)$ time,
  \end{itemize}
  the algorithm runs in $\Oh(n(P(n) + n(I(n)+D(n))))$ time. As we
  argue next, for (weighted) points in $\R^d$, there is such a data
  structure with $P(n)=\Oh(n\log^d n)$, and
  $I(n) \leq D(n)=\Oh(\log^{d+2} n)$. The theorem then follows.

  We describe a dynamic data structure for weighted nearest neighbor
  (NN) queries and then apply a recent result of
  Chan~\cite{chan20dynam_gener_closes_pair} to turn this into a
  dynamic bichromatic closest pair data structure.

  Let $P'$ be a set of weighted points in $\R^d$, let $q$ be a
  weighted point, and consider the positive orthant
  $A(q) = \{ a \mid q_i \leq a_i \text{ for all } i \in \{1,\ldots,d\}\}
  \subset \R^d$ with respect to $q$ (i.e. all points dominating
  $q$). Let $w'_z = -w_{p'} + \sum_{i=1}^d p'_i$ (see
  Figure~\ref{fig:weighted_nn_ds}).  Now observe that the point
  $p^* \in P' \cap A(q)$ minimizing $\|p'-q\| - w_{p'} - w_q$ is the
  point in $A(q)$ with minimum $w'_{p'}$ value.

  \begin{figure}[tb]
    \centering
    \includegraphics{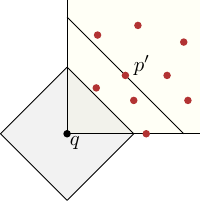}
    \caption{An example of a $L_1$ unit ball centered at $q$ in
      $\R^2$. Observe that all points on the line through $p'$ with
      slope minus one have the same value $\sum_{i=1} p'_i$. For a
      weighted point $p' \in P'$ that lies in $A(q)$ (the yellow
      region), the weighted $L_1$-distance between $q$ and $p'$ is
      given by $w'_p - w'_q$, hence the point from $P' \cap A(q)$ with
      minimum $w'_p$ value is the nearest neighbor of $q$.  }
    \label{fig:weighted_nn_ds}
  \end{figure}

  So, we store the points in $P'$ in a $d$-dimension range
  tree~\cite{preparata2012computational}, in which every subtree is
  annotated with the point with the minimum $w'_p$ value among its
  descendants.  We can query the range tree for the nearest point
  among $P'\cap A(q)$ with respect to the weighted $L_1$ distance in
  $\Oh(\log^d n)$ time. By implementing the trees using generalized
  balanced trees~\cite{andersson99gener_balan_trees} we can support
  insertions and deletions in amortized $\Oh(\log^d n)$ time as
  well. We use an analogous approach for the other $2^d-1$ orthants
  around $q$, thus allowing us to answer weighted NN queries with
  respect to the $L_1$ distance in the same time as above. We can
  handle the (weighted) $L_\infty$ distance analogously by defining
  appropriate ``orthants'' around $q$. If $P'$ contains any dummy
  points, we store them separately in a binary search tree, ordered by
  weight. The minimum weight dummy point is an additional candidate
  nearest neighbor for $q$, and we return the overall closest point to
  $q$. When $q$ itself would be a dummy point (and thus has no real
  location), we simply obtain the minimum weight point among $P'$ as
  the nearest point. Hence, it follows we can build the structure in
  $P_0(n)=O(n\log^d n)$ time, query the nearest neighbor in
  $Q_0(n)=O(\log^d n)$ time, and delete (and insert) points in
  $D_0(n)=O(\log^d n)$ time.

  Chan~\cite{chan20dynam_gener_closes_pair} shows how to turn a
  dynamic nearest neighbor searching data structure with preprocessing
  time $P_0(n)$, query time $Q_0(n)$, and deletion time $D_0(n)$ into
  a fully dynamic bichromatic closest pair data structure. In our
  particular case we obtain
  \begin{itemize}[noitemsep]
  \item construction time $P(n)=\Oh(nQ_0(n)+P_0(n))=\Oh(n\log^d n)$,
  \item amortized insertion time $I(n)=\Oh(Q_0(n)\log n + (P_0(n)/n)\log
    n)=\Oh(\log^{d+1}n)$, and
  \item amortized deletion time $D(n)=\Oh(Q_0(n)\log^2 n +
    (P_0(n)/n)\log^2 n + D_0(n)\log n) = \Oh(\log^{d+2} n)$.
  \end{itemize}
  Plugging this into Vaidia's algorithm we thus obtain an
  $\Oh(n^2\log^{d+2}n)$ time algorithm to compute
  $\EMD(B,R)$.
\end{proof}

\paragraph{Earth Mover's Distance under Translation in $L_1$.} The
sets $B$ and $R$ are aligned in dimension~$i$, or \emph{$i$-aligned}
for short, if there is a pair of points $b \in B, r \in R$ for which
$b_i=r_i$. Eppstein
\etal~\cite{eppstein15improv_grid_map_layout_point_set_match} show
that for two point sets in $\R^2$, there exists an optimal translation
$\tau^*$ that aligns $B$ and $R$ in both dimensions. They explicitly
consider all $\Oh((nm)^2)$ translations that both $1$-align and
$2$-align $B+\tau$ and $R$. For each such a translation $\tau$,
computing an optimal matching can then be done in $\Oh(n^2\log^3 n)$
time~\cite{vaidya89geomet_helps_match}, thus leading to an
$\Oh(n^4 m^2\log^3 n)$ time algorithm. We now argue that we can
generalize the above result to higher dimensions.

\begin{restatable}{theorem}{l1higherDEmdUnderTranslationAlgorithm}
  \label{thm:l1_solution}
  Given $B$ and $R$ we can find an optimal translation $\tau^*$
  realizing $\EMDuT_1(B,R)$ in $\Oh(m^dn^{d+2}\log^{d+2}n)$ time.
\end{restatable}

\begin{proof}
  Recall the definition of the cost function
  \[
    \D_{B,R,1}(\phi,\tau) =
    \sum_{b \in B} L_1(b+\tau,\phi(b)) = \sum_{b \in B}\sum_{i=1}^d
    |b_i + \tau_i - \phi(b)_i|.
  \]
  For a fixed matching $\phi$, this is a piecewise linear function in
  $\tau$. In particular, $\D_{B,R,1}(\phi,\tau)$ is a sum of piecewise
  linear functions $f_{b,i}(\tau) = |b_i+\tau_i-\phi(b)_i|$. For each
  such a function there is a hyperplane $h_{b,\phi(b),i}$ in $\R^d$
  given by the equation $\tau_i + b_i -\phi(b)_i = 0$, so that for a
  point (translation) $\tau \in \R^d$ on one side of (or on) the
  hyperplane, $f_{b,i}(\tau)$ is linear in $\tau$ (i.e. on
  one side we have $f(\tau) = \tau_i + b_i -\phi(b)_i$, whereas on the
  other side we have $f(\tau) = -\tau_i - b_i +\phi(b)_i$). Let
  $H_\phi = \{h_{b,\phi(b),i} \mid b \in B, i \in \{1,\ldots,d\}\}$ denote the
  set of all such hyperplanes, and consider the arrangement
  $\A(H_\phi)$. It follows that in each cell of $\A(H_\phi)$, the
  function $\D_{B,R,1}(\phi,\tau)$ is a linear function in $\tau$, and that
  $\D_{B,R,1}(\phi,\tau)$ thus has its minimum at a vertex of
  $\A(H_\phi)$.

  We extend the set of hyperplanes $H_\phi$ to include the hyperplane
  $h_{b,r,i}$ for \emph{every} pair $(b,r) \in B\times R$, and every
  $i\in \{1,\ldots,d\}$, rather than just the pairs $(b,\phi(b))$. Let $H$ be
  the resulting set. A minimum of $\D_{B,R,1}(\phi,\tau)$
  still occurs at a vertex of $\A(H)$ (as $\A(H)$ includes
  all vertices of $\A(H_\phi)$). Moreover, observe that $H$ now
  actually contains the hyperplanes $H_\phi$, for \emph{every}
  matching $\phi \in \Phi$, so also those of an optimal matching
  $\phi^*$. It thus follows that such a global minimum
  $\D_1(\phi^*,\tau^*)$ occurs at a vertex $\tau^*$ of $\A(H)$.

  So, to compute an optimal matching $\phi^*$ and its $\tau^*$ (and thus $\EMDuT(B,R)$) we
  can
  \begin{enumerate}[noitemsep]
  \item explicitly compute (all vertices of) $\A(H)$,
  \item for each such a vertex $\tau \in \A(H)$ (which is some
    candidate translation), compute an optimal matching $\phi_\tau$ between the sets
    $B + \tau$ and $R$, and
  \item report the matching (and corresponding translation) that minimizes total
      cost.
  \end{enumerate}

  The set $H$ contains $mnd$ hyperplanes, and thus $\A(H)$ contains
  $\Oh((mnd)^d)=\Oh(m^dn^{d})$ vertices. Computing $\A(H)$ takes
  $\Oh(m^dn^d)$
  time~\cite{bkos-cgaa-08,edelsbrunner86const_arran_lines_hyper_applic}. For
  each such a vertex (translation), we can compute an optimal matching
  in $\Oh(n^2\log^{d+2}n))$ time using the algorithm from
  Theorem~\ref{thm:l1_fixed_sets}. This thus yields an
  $\Oh(m^dn^{d+2}\log^{d+2}n)$ time algorithm in total.
\end{proof}

\paragraph{EMD under translation in $L_\infty$.} 
We present the following analog for $L_\infty$.

\begin{restatable}{theorem}{lInftyHigherDimAlgorithm}
  \label{thm:linfty_solution}
  Given $B$ and $R$ we can find an optimal translation $\tau^*$
  realizing $\EMDuT_\infty(B,R)$ in
  $\Oh(m^dn^{d+2}\log^{d+2}n)$ time.
\end{restatable}

\begin{proof}
  We use a similar approach as in Theorem~\ref{thm:l1_solution};
  i.e. we prove that there is a set $H$ of $\Oh(mnd^2)$ hyperplanes in
  $\R^d$, so that for any matching $\phi$, there is a minimum cost
  translation that is a vertex of the arrangement $\A(H)$. We can thus
  again compute such an optimal matching (and the translation) by
  trying all $\Oh(m^dn^d)$ vertices. This yields an
  $\Oh(m^dn^{d+2}\log^{d+2}n)$ time algorithm
  as claimed. What remains is to describe the set of hyperplanes $H$.

  Fix a matching $\phi$. We have
  \[
    \D_{B,R,\infty}(\phi,\tau) = \sum_{b \in B}
    L_\infty(b+\tau,\phi(b)) = \sum_{b \in B} \max_{i=1}^d |b_i +
    \tau_i - \phi(b)_i|,
  \]
  which is again a piecewise linear function in $\tau$, as it is a sum
  of piecewise linear functions. Each individual function is now of
  the form $f_{b,r} = \max_{i=1}^d |b_i + \tau_i - r_i|$, where
  $b \in B$, and $r=\phi(b)$. For each such a function, there are now
  at most $\Oh(d^2)$ hyperplanes that subdivide $\R^d$ into regions in
  which $f_{b,r}$ is given by a single linear function. In particular,
  the $d$ hyperplanes given by $\tau_i-b_i+r_i=0$, for any
  $i \in \{1,\ldots,d\}$, and $\Oh(d^2)$ hyperplanes that we get from solving
  $|b_i+\tau_i-r_i|=|b_j+\tau_j-r_j|$ for each $j \neq i$ (for example
  $\tau_i- \tau_j + b_i - b_j - r_i + r_j=0$). Let $H$ denote the
  resulting set of $\Oh(mnd^2)$ hyperplanes over all points $b \in B$,
  $r \in R$, and all $i \in \{1,\ldots,d\}$. It again follows that for any
  matching $\phi$, there is a vertex of $\A(H)$ that corresponds to a
  translation $\tau_\phi$ that minimizes
  $\D_{B,R,\infty}(\phi,\tau)$. Since this holds for an optimal
  matching $\phi^*$, we can thus compute $\EMDuT(B,R)$ in
  $\Oh(m^dn^{d+2}\log^{d+2}n)$ time.
\end{proof}


\bibliographystyle{plain}
\bibliography{bib}

@inproceedings{icm-survey,
	author = { Vassilevska-Williams, Virginia},
	title = {{O}n {s}ome {f}ine-{g}rained {q}uestions in {a}lgorithms and {c}omplexity},
	booktitle = {International Congress of Mathematicians (ICM 2018)},
	chapter = {},
	year	= {2018},
	pages = {3447-34},
	doi = {10.1142/9789813272880_0188}
}

@article{ov-seth,
    author    = {Ryan Williams},
    title     = {A new algorithm for optimal 2-constraint satisfaction and its implications},
    journal   = {Theoretical Computer Science},
    volume    = {348},
    number    = {2-3},
    pages     = {357--365},
    year      = {2005},
    *url       = {https://doi.org/10.1016/j.tcs.2005.09.023},
    doi       = {10.1016/j.tcs.2005.09.023},
    timestamp = {Wed, 17 Feb 2021 21:59:05 +0100},
    biburl    = {https://dblp.org/rec/journals/tcs/Williams05.bib},
    bibsource = {dblp computer science bibliography, https://dblp.org}
}

@article{eppstein15improv_grid_map_layout_point_set_match,
  author       = {David Eppstein and
                  Marc J. van Kreveld and
                  Bettina Speckmann and
                  Frank Staals},
  title        = {Improved Grid Map Layout by Point Set Matching},
  journal      = {International Journal of Computational Geometry \& Applications},
  volume       = 25,
  number       = 2,
  pages        = {101--122},
  year         = 2015,
  *url          = {https://doi.org/10.1142/S0218195915500077},
  doi          = {10.1142/S0218195915500077},
  timestamp    = {Fri, 09 Apr 2021 18:33:04 +0200},
  biburl       = {https://dblp.org/rec/journals/ijcga/EppsteinKSS15.bib},
  bibsource    = {dblp computer science bibliography, https://dblp.org}
}

@article{vaidya89geomet_helps_match,
  author       = {Pravin M. Vaidya},
  title        = {Geometry Helps in Matching},
  journal      = {{SIAM} Journal on Computing},
  volume       = 18,
  number       = 6,
  pages        = {1201--1225},
  year         = 1989,
  *url          = {https://doi.org/10.1137/0218080},
  doi          = {10.1137/0218080},
  timestamp    = {Sat, 27 May 2017 14:22:59 +0200},
  biburl       = {https://dblp.org/rec/journals/siamcomp/Vaidya89a.bib},
  bibsource    = {dblp computer science bibliography, https://dblp.org}
}

@inproceedings{Rohatgi19,
  author       = {Dhruv Rohatgi},
  *editor       = {Dimitris Achlioptas and
                  L{\'{a}}szl{\'{o}} A. V{\'{e}}gh},
  title        = {{Conditional hardness of Earth Mover distance}},
  booktitle    = {Approximation, Randomization, and Combinatorial Optimization. Algorithms
                  and Techniques ({APPROX/RANDOM} 2019)},
  series       = {LIPIcs},
  volume       = {145},
  pages        = {12:1--12:17},
  publisher    = {Schloss Dagstuhl~-- Leibniz-Zentrum f{\"{u}}r Informatik},
  year         = {2019},
  *url          = {https://doi.org/10.4230/LIPIcs.APPROX-RANDOM.2019.12},
  doi          = {10.4230/LIPICS.APPROX-RANDOM.2019.12},
  timestamp    = {Tue, 21 Sep 2021 09:36:24 +0200},
  biburl       = {https://dblp.org/rec/conf/approx/Rohatgi19.bib},
  bibsource    = {dblp computer science bibliography, https://dblp.org}
}

@inproceedings{AndoniBIW09,
  author       = {Alexandr Andoni and
                  Khanh Do Ba and
                  Piotr Indyk and
                  David P. Woodruff},
  title        = {{Efficient sketches for Earth-Mover distance, with applications}},
  booktitle    = {50th Annual Symposium on Foundations of Computer Science ({FOCS}
                  2009)},
  pages        = {324--330},
  publisher    = {{IEEE} Computer Society},
  year         = {2009},
  *url          = {https://doi.org/10.1109/FOCS.2009.25},
  doi          = {10.1109/FOCS.2009.25},
  timestamp    = {Thu, 23 Mar 2023 23:57:52 +0100},
  biburl       = {https://dblp.org/rec/conf/focs/AndoniBIW09.bib},
  bibsource    = {dblp computer science bibliography, https://dblp.org}
}

@inproceedings{AndoniIK08,
  author       = {Alexandr Andoni and
                  Piotr Indyk and
                  Robert Krauthgamer},
  *editor       = {Shang{-}Hua Teng},
  title        = {{Earth mover distance over high-dimensional spaces}},
  booktitle    = {19th Annual Symposium on Discrete
                  Algorithms ({SODA} 2008)},
  pages        = {343--352},
  publisher    = {{SIAM}},
  year         = {2008},
  url          = {http://dl.acm.org/citation.cfm?id=1347082.1347120},
  timestamp    = {Fri, 07 Dec 2012 17:02:08 +0100},
  biburl       = {https://dblp.org/rec/conf/soda/AndoniIK08.bib},
  bibsource    = {dblp computer science bibliography, https://dblp.org}
}

@article{RubnerTG00,
  author       = {Yossi Rubner and
                  Carlo Tomasi and
                  Leonidas J. Guibas},
  title        = {{The Earth Mover's} distance as a metric for image retrieval},
  journal      = { International Journal of Computer Vision},
  volume       = {40},
  number       = {2},
  pages        = {99--121},
  year         = {2000},
  *url          = {https://doi.org/10.1023/A:1026543900054},
  doi          = {10.1023/A:1026543900054},
  timestamp    = {Fri, 13 Mar 2020 10:59:09 +0100},
  biburl       = {https://dblp.org/rec/journals/ijcv/RubnerTG00.bib},
  bibsource    = {dblp computer science bibliography, https://dblp.org}
}

@article{ChenHKX06,
  author       = {Jianer Chen and
                  Xiuzhen Huang and
                  Iyad A. Kanj and
                  Ge Xia},
  title        = {Strong computational lower bounds via parameterized complexity},
  journal      = {Journal of Computer and System Sciences},
  volume       = {72},
  number       = {8},
  pages        = {1346--1367},
  year         = {2006},
  *url          = {https://doi.org/10.1016/j.jcss.2006.04.007},
  doi          = {10.1016/J.JCSS.2006.04.007},
  timestamp    = {Tue, 16 Feb 2021 14:03:39 +0100},
  biburl       = {https://dblp.org/rec/journals/jcss/ChenHKX06.bib},
  bibsource    = {dblp computer science bibliography, https://dblp.org}
}

@inproceedings{BringmannKN20,
  author       = {Karl Bringmann and
                  Marvin K{\"{u}}nnemann and
                  Andr{\'{e}} Nusser},
  *editor       = {Fabrizio Grandoni and
                  Grzegorz Herman and
                  Peter Sanders},
  title        = {{When Lipschitz walks your dog: Algorithm engineering of the discrete
                  Fr{\'{e}}chet distance under translation}},
  booktitle    = {28th Annual European Symposium on Algorithms ({ESA} 2020)},
  series       = {LIPIcs},
  volume       = {173},
  pages        = {25:1--25:17},
  publisher    = {Schloss Dagstuhl~-- Leibniz-Zentrum f{\"{u}}r Informatik},
  year         = {2020},
  *url          = {https://doi.org/10.4230/LIPIcs.ESA.2020.25},
  doi          = {10.4230/LIPICS.ESA.2020.25},
  timestamp    = {Mon, 21 Dec 2020 13:23:22 +0100},
  biburl       = {https://dblp.org/rec/conf/esa/BringmannKN20.bib},
  bibsource    = {dblp computer science bibliography, https://dblp.org}
}

@inproceedings{AgarwalRSS22,
  author       = {Pankaj K. Agarwal and
                  Sharath Raghvendra and
                  Pouyan Shirzadian and
                  Rachita Sowle},
  *editor       = {Artur Czumaj and
                  Qin Xin},
  title        = {An Improved {\(\epsilon\)}-Approximation Algorithm for Geometric Bipartite
                  Matching},
  booktitle    = {18th Scandinavian Symposium and Workshops on Algorithm Theory ({SWAT}
                  2022)},
  series       = {LIPIcs},
  volume       = {227},
  pages        = {6:1--6:20},
  publisher    = {Schloss Dagstuhl~-- Leibniz-Zentrum f{\"{u}}r Informatik},
  year         = {2022},
  *url          = {https://doi.org/10.4230/LIPIcs.SWAT.2022.6},
  doi          = {10.4230/LIPICS.SWAT.2022.6},
  timestamp    = {Thu, 23 Jun 2022 15:53:59 +0200},
  biburl       = {https://dblp.org/rec/conf/swat/AgarwalRSS22.bib},
  bibsource    = {dblp computer science bibliography, https://dblp.org}
}

@inproceedings{AgarwalCRX22,
  author       = {Pankaj K. Agarwal and
                  Hsien{-}Chih Chang and
                  Sharath Raghvendra and
                  Allen Xiao},
  *editor       = {Stefano Leonardi and
                  Anupam Gupta},
  title        = {Deterministic, near-linear \emph{{\(\epsilon\)}}-approximation algorithm
                  for geometric bipartite matching},
  booktitle    = {54th Annual {ACM} {SIGACT} Symposium on Theory of Computing (STOC 2022)},
  pages        = {1052--1065},
  publisher    = {{ACM}},
  year         = {2022},
  *url          = {https://doi.org/10.1145/3519935.3519977},
  doi          = {10.1145/3519935.3519977},
  timestamp    = {Tue, 27 Dec 2022 09:06:31 +0100},
  biburl       = {https://dblp.org/rec/conf/stoc/AgarwalCRX22.bib},
  bibsource    = {dblp computer science bibliography, https://dblp.org}
}

@inproceedings{FoxL23,
  author       = {Emily Fox and
                  Jiashuai Lu},
  title        = {A deterministic near-linear time approximation scheme for geometric
                  transportation},
  booktitle    = {64th Annual Symposium on Foundations of Computer Science ({FOCS}
                  2023)},
  pages        = {1301--1315},
  publisher    = {{IEEE}},
  year         = 2023,
  *url          = {https://doi.org/10.1109/FOCS57990.2023.00078},
  doi          = {10.1109/FOCS57990.2023.00078},
  timestamp    = {Tue, 02 Jan 2024 15:09:54 +0100},
  biburl       = {https://dblp.org/rec/conf/focs/FoxL23.bib},
  bibsource    = {dblp computer science bibliography, https://dblp.org}
}

@article{FoxL22,
  author       = {Kyle Fox and
                  Jiashuai Lu},
  title        = {A near-linear time approximation scheme for geometric transportation
                  with arbitrary supplies and spread},
  journal      = {Journal of Computational Geometry},
  volume       = {13},
  number       = {1},
  year         = {2022},
  *url          = {https://doi.org/10.20382/jocg.v13i1a8},
  doi          = {10.20382/JOCG.V13I1A8},
  timestamp    = {Mon, 20 Jun 2022 17:38:52 +0200},
  biburl       = {https://dblp.org/rec/journals/jocg/FoxL22.bib},
  bibsource    = {dblp computer science bibliography, https://dblp.org},
  pages        = {204--225}
}

@inproceedings{AndoniNOY14,
  author       = {Alexandr Andoni and
                  Aleksandar Nikolov and
                  Krzysztof Onak and
                  Grigory Yaroslavtsev},
  *editor       = {David B. Shmoys},
  title        = {Parallel algorithms for geometric graph problems},
  booktitle    = {Symposium on Theory of Computing, ({STOC} 2014)},
  pages        = {574--583},
  publisher    = {{ACM}},
  year         = {2014},
  *url          = {https://doi.org/10.1145/2591796.2591805},
  doi          = {10.1145/2591796.2591805},
  timestamp    = {Sun, 02 Oct 2022 16:16:11 +0200},
  biburl       = {https://dblp.org/rec/conf/stoc/AndoniNOY14.bib},
  bibsource    = {dblp computer science bibliography, https://dblp.org}
}

@inproceedings{Indyk07,
  author       = {Piotr Indyk},
  *editor       = {Nikhil Bansal and
                  Kirk Pruhs and
                  Clifford Stein},
  title        = {A near linear time constant factor approximation for {E}uclidean bichromatic
                  matching (cost)},
  booktitle    = {18th Annual Symposium on Discrete
                  Algorithms ({SODA} 2007)},
  pages        = {39--42},
  publisher    = {{SIAM}},
  year         = {2007},
  url          = {http://dl.acm.org/citation.cfm?id=1283383.1283388},
  timestamp    = {Tue, 15 Feb 2022 07:54:27 +0100},
  biburl       = {https://dblp.org/rec/conf/soda/Indyk07.bib},
  bibsource    = {dblp computer science bibliography, https://dblp.org}
}

@article{KhesinNP20,
  author       = {Andrey Boris Khesin and
                  Aleksandar Nikolov and
                  Dmitry Paramonov},
  title        = {Preconditioning for the Geometric Transportation Problem},
  journal      = {Journal of Computational Geometry},
  volume       = {11},
  number       = {2},
  pages        = {234--259},
  year         = {2020},
  *url          = {https://doi.org/10.20382/jocg.v11i2a11},
  doi          = {10.20382/JOCG.V11I2A11},
  timestamp    = {Mon, 09 May 2022 16:20:13 +0200},
  biburl       = {https://dblp.org/rec/journals/jocg/KhesinNP20.bib},
  bibsource    = {dblp computer science bibliography, https://dblp.org}
}

@article{JiangXZ08,
  author       = {Minghui Jiang and
                  Ying Xu and
                  Binhai Zhu},
  title        = {Protein Structure-structure Alignment with Discrete Fr{\'{e}}chet
                  Distance},
  journal      = {Journal of Bioinformatics and Computational Biology},
  volume       = {6},
  number       = {1},
  pages        = {51--64},
  year         = {2008},
  *url          = {https://doi.org/10.1142/S0219720008003278},
  doi          = {10.1142/S0219720008003278},
  timestamp    = {Thu, 04 Jun 2020 19:43:04 +0200},
  biburl       = {https://dblp.org/rec/journals/jbcb/JiangXZ08.bib},
  bibsource    = {dblp computer science bibliography, https://dblp.org}
}

@article{MosigC05,
  author       = {Axel Mosig and
                  Michael Clausen},
  title        = {Approximately matching polygonal curves with respect to the {F}r{\'{e}}chet
                  distance},
  journal      = {Computational Geometry},
  volume       = {30},
  number       = {2},
  pages        = {113--127},
  year         = {2005},
  *url          = {https://doi.org/10.1016/j.comgeo.2004.05.004},
  doi          = {10.1016/J.COMGEO.2004.05.004},
  timestamp    = {Thu, 11 Feb 2021 23:27:29 +0100},
  biburl       = {https://dblp.org/rec/journals/comgeo/MosigC05.bib},
  bibsource    = {dblp computer science bibliography, https://dblp.org}
}

@inproceedings{AltKW01,
  author       = {Helmut Alt and
                  Christian Knauer and
                  Carola Wenk},
  *editor       = {Afonso Ferreira and
                  Horst Reichel},
  title        = {Matching Polygonal Curves with Respect to the {F}r{\'{e}}chet Distance},
  booktitle    = {18th Annual Symposium on Theoretical Aspects of Computer Science {(STACS 2001)}},
  series       = {Lecture Notes in Computer Science},
  volume       = {2010},
  pages        = {63--74},
  publisher    = {Springer},
  year         = {2001},
  *url          = {https://doi.org/10.1007/3-540-44693-1\_6},
  doi          = {10.1007/3-540-44693-1_6},
  timestamp    = {Tue, 14 May 2019 10:00:48 +0200},
  biburl       = {https://dblp.org/rec/conf/stacs/AltKW01.bib},
  bibsource    = {dblp computer science bibliography, https://dblp.org}
}

@article{BringmannKN21,
  author       = {Karl Bringmann and
                  Marvin K{\"{u}}nnemann and
                  Andr{\'{e}} Nusser},
  title        = {{Discrete Fr{\'{e}}chet distance under translation: Conditional
                  hardness and an improved algorithm}},
  journal      = {{ACM} Trans. Algorithms},
  volume       = {17},
  number       = {3},
  pages        = {25:1--25:42},
  year         = {2021},
  *url          = {https://doi.org/10.1145/3460656},
  doi          = {10.1145/3460656},
  timestamp    = {Thu, 12 Aug 2021 08:58:06 +0200},
  biburl       = {https://dblp.org/rec/journals/talg/BringmannKN21.bib},
  bibsource    = {dblp computer science bibliography, https://dblp.org}
}

@article{FiltserK20,
  author       = {Omrit Filtser and
                  Matthew J. Katz},
  title        = {{Algorithms for the discrete Fr{\'{e}}chet distance under translation}},
  journal      = {Journal of Computational Geometry},
  volume       = {11},
  number       = {1},
  pages        = {156--175},
  year         = {2020},
  *url          = {https://doi.org/10.20382/jocg.v11i1a7},
  doi          = {10.20382/JOCG.V11I1A7},
  timestamp    = {Mon, 09 May 2022 16:20:13 +0200},
  biburl       = {https://dblp.org/rec/journals/jocg/FiltserK20.bib},
  bibsource    = {dblp computer science bibliography, https://dblp.org}
}

@article{AvrahamKS15,
  author       = {Rinat Ben Avraham and
                  Haim Kaplan and
                  Micha Sharir},
  title        = {{A faster algorithm for the discrete Fr{\'{e}}chet distance under
                  translation}},
  journal      = {CoRR},
  volume       = {abs/1501.03724},
  year         = {2015},
  *url          = {http://arxiv.org/abs/1501.03724},
  eprinttype    = {arXiv},
  eprint       = {1501.03724},
  timestamp    = {Mon, 13 Aug 2018 16:46:45 +0200},
  biburl       = {https://dblp.org/rec/journals/corr/AvrahamKS15.bib},
  bibsource    = {dblp computer science bibliography, https://dblp.org}
}

@inproceedings{Chan23,
  author       = {Timothy M. Chan},
  title        = {{Minimum {$L_\infty$} Hausdorff distance of point sets under translation:
                  Generalizing Klee's measure problem}},
  booktitle    = {39th International Symposium on Computational Geometry (SoCG 2023)},
  series       = {LIPIcs},
  volume       = {258},
  pages        = {24:1--24:13},
  publisher    = {Schloss Dagstuhl~-- Leibniz-Zentrum f{\"{u}}r Informatik},
  year         = {2023},
  *url          = {https://doi.org/10.4230/LIPIcs.SoCG.2023.24},
  doi          = {10.4230/LIPICS.SOCG.2023.24},
  timestamp    = {Tue, 13 Jun 2023 16:59:20 +0200},
  biburl       = {https://dblp.org/rec/conf/compgeom/Chan23.bib},
  bibsource    = {dblp computer science bibliography, https://dblp.org}
}

@inproceedings{BringmannN21,
  author       = {Karl Bringmann and
                  Andr{\'{e}} Nusser},
  *editor       = {Kevin Buchin and
                  {\'{E}}ric Colin de Verdi{\`{e}}re},
  title        = {{Translating Hausdorff is hard: Fine-grained lower bounds for Hausdorff
                  distance under translation}},
  booktitle    = {37th International Symposium on Computational Geometry (SoCG 2021)},
  series       = {LIPIcs},
  volume       = {189},
  pages        = {18:1--18:17},
  publisher    = {Schloss Dagstuhl~-- Leibniz-Zentrum f{\"{u}}r Informatik},
  year         = {2021},
  *url          = {https://doi.org/10.4230/LIPIcs.SoCG.2021.18},
  doi          = {10.4230/LIPICS.SOCG.2021.18},
  timestamp    = {Fri, 04 Jun 2021 19:47:03 +0200},
  biburl       = {https://dblp.org/rec/conf/compgeom/BringmannN21.bib},
  bibsource    = {dblp computer science bibliography, https://dblp.org}
}

@article{KnauerKS11,
  author       = {Christian Knauer and
                  Klaus Kriegel and
                  Fabian Stehn},
  title        = {{Minimizing the weighted directed Hausdorff distance between colored
                  point sets under translations and rigid motions}},
  journal      = {Theoretical Computer Science},
  volume       = {412},
  number       = {4-5},
  pages        = {375--382},
  year         = {2011},
  *url          = {https://doi.org/10.1016/j.tcs.2010.03.020},
  doi          = {10.1016/J.TCS.2010.03.020},
  timestamp    = {Wed, 17 Feb 2021 21:57:01 +0100},
  biburl       = {https://dblp.org/rec/journals/tcs/KnauerKS11.bib},
  bibsource    = {dblp computer science bibliography, https://dblp.org}
}

@article{KnauerS11,
  author       = {Christian Knauer and
                  Marc Scherfenberg},
  title        = {{Approximate nearest neighbor search under translation invariant Hausdorff
                  distance}},
  journal      = {International Journal of Computational Geometry \& Applications},
  volume       = {21},
  number       = {3},
  pages        = {369--381},
  year         = {2011},
  *url          = {https://doi.org/10.1142/S0218195911003706},
  doi          = {10.1142/S0218195911003706},
  timestamp    = {Thu, 04 Jun 2020 19:43:31 +0200},
  biburl       = {https://dblp.org/rec/journals/ijcga/KnauerS11.bib},
  bibsource    = {dblp computer science bibliography, https://dblp.org}
}

@article{AgarwalHSW10,
  author       = {Pankaj K. Agarwal and
                  Sariel Har{-}Peled and
                  Micha Sharir and
                  Yusu Wang},
  title        = {{Hausdorff distance under translation for points and balls}},
  journal      = {{ACM} Trans. Algorithms},
  volume       = {6},
  number       = {4},
  pages        = {71:1--71:26},
  year         = {2010},
  *url          = {https://doi.org/10.1145/1824777.1824791},
  doi          = {10.1145/1824777.1824791},
  timestamp    = {Mon, 02 Jan 2023 09:02:12 +0100},
  biburl       = {https://dblp.org/rec/journals/talg/AgarwalHSW10.bib},
  bibsource    = {dblp computer science bibliography, https://dblp.org}
}

@inproceedings{HuttenlocherRK92,
  author       = {Daniel P. Huttenlocher and
                  William Rucklidge and
                  Gregory A. Klanderman},
  title        = {{Comparing images using the Hausdorff distance under translation}},
  booktitle    = {{Computer Society Conference on Computer Vision and Pattern
                  Recognition ({CVPR} 1992)}},
  pages        = {654--656},
  publisher    = {{IEEE}},
  year         = {1992},
  *url          = {https://doi.org/10.1109/CVPR.1992.223209},
  doi          = {10.1109/CVPR.1992.223209},
  timestamp    = {Wed, 16 Oct 2019 14:14:50 +0200},
  biburl       = {https://dblp.org/rec/conf/cvpr/HuttenlocherRK92.bib},
  bibsource    = {dblp computer science bibliography, https://dblp.org}
}

@article{Rote91,
  author       = {G{\"{u}}nter Rote},
  title        = {{Computing the minimum Hausdorff distance between two point sets on
                  a line under translation}},
  journal      = {Information Processing Letters},
  volume       = {38},
  number       = {3},
  pages        = {123--127},
  year         = {1991},
  *url          = {https://doi.org/10.1016/0020-0190(91)90233-8},
  doi          = {10.1016/0020-0190(91)90233-8},
  timestamp    = {Fri, 26 May 2017 22:54:46 +0200},
  biburl       = {https://dblp.org/rec/journals/ipl/Rote91.bib},
  bibsource    = {dblp computer science bibliography, https://dblp.org}
}

@inproceedings{HuttenlocherK90,
  author       = {Daniel P. Huttenlocher and
                  Klara Kedem},
  *editor       = {Raimund Seidel},
  title        = {{Computing the minimum Hausdorff distance for point sets under translation}},
  booktitle    = {6th Annual Symposium on Computational Geometry (SoCG 1990)},
  pages        = {340--349},
  publisher    = {{ACM}},
  year         = {1990},
  *url          = {https://doi.org/10.1145/98524.98599},
  doi          = {10.1145/98524.98599},
  timestamp    = {Mon, 14 Jun 2021 16:24:55 +0200},
  biburl       = {https://dblp.org/rec/conf/compgeom/HuttenlocherK90.bib},
  bibsource    = {dblp computer science bibliography, https://dblp.org}
}

@article{CabelloGKR08,
  author       = {Sergio Cabello and
                  Panos Giannopoulos and
                  Christian Knauer and
                  G{\"{u}}nter Rote},
  title        = {{Matching point sets with respect to the Earth Mover's distance}},
  journal      = {Comput. Geom.},
  volume       = {39},
  number       = {2},
  pages        = {118--133},
  year         = {2008},
  *url          = {https://doi.org/10.1016/j.comgeo.2006.10.001},
  doi          = {10.1016/J.COMGEO.2006.10.001},
  timestamp    = {Thu, 11 Feb 2021 23:27:16 +0100},
  biburl       = {https://dblp.org/rec/journals/comgeo/CabelloGKR08.bib},
  bibsource    = {dblp computer science bibliography, https://dblp.org}
}

@inproceedings{KleinV05,
  author       = {Oliver Klein and
                  Remco C. Veltkamp},
  title        = {{Approximation algorithms for the Earth Mover's distance under transformations
                  using reference points}},
  booktitle    = {21st European Workshop on Computational
                  Geometry (EWCG 2005)},
  pages        = {53--56},
  publisher    = {Technische Universiteit Eindhoven},
  year         = {2005},
  url          = {http://www.win.tue.nl/EWCG2005/Proceedings/14.pdf},
  bibsource    = {dblp computer science bibliography, https://dblp.org}
}

@inproceedings{CohenG99,
  author       = {Scott D. Cohen and
                  Leonidas J. Guibas},
  title        = {{The Earth Mover's distance under transformation sets}},
  booktitle    = {International Conference on Computer Vision (ICCV 1999)},
  pages        = {1076--1083},
  publisher    = {{IEEE} Computer Society},
  year         = {1999},
  *url          = {https://doi.org/10.1109/ICCV.1999.790393},
  doi          = {10.1109/ICCV.1999.790393},
  timestamp    = {Thu, 23 Mar 2023 23:57:41 +0100},
  biburl       = {https://dblp.org/rec/conf/iccv/CohenG99.bib},
  bibsource    = {dblp computer science bibliography, https://dblp.org}
}

@article{ImpagliazzoP01,
  author       = {Russell Impagliazzo and
                  Ramamohan Paturi},
  title        = {{On the complexity of {$k$}-SAT}},
  journal      = {Journal of Computer and System Sciences},
  volume       = {62},
  number       = {2},
  pages        = {367--375},
  year         = {2001},
  *url          = {https://doi.org/10.1006/jcss.2000.1727},
  doi          = {10.1006/JCSS.2000.1727},
  timestamp    = {Tue, 16 Feb 2021 14:04:38 +0100},
  biburl       = {https://dblp.org/rec/journals/jcss/ImpagliazzoP01.bib},
  bibsource    = {dblp computer science bibliography, https://dblp.org}
}

@article{Bajaj88,
  author       = {Chandrajit L. Bajaj},
  title        = {The Algebraic Degree of Geometric Optimization Problems},
  journal      = {Discrete Computational Geometry},
  volume       = {3},
  pages        = {177--191},
  year         = {1988},
  *url          = {https://doi.org/10.1007/BF02187906},
  doi          = {10.1007/BF02187906},
  timestamp    = {Thu, 12 Mar 2020 17:20:43 +0100},
  biburl       = {https://dblp.org/rec/journals/dcg/Bajaj88.bib},
  bibsource    = {dblp computer science bibliography, https://dblp.org}
}

@inproceedings{CohenLMPS16,
  author       = {Michael B. Cohen and
                  Yin Tat Lee and
                  Gary L. Miller and
                  Jakub Pachocki and
                  Aaron Sidford},
  *editor       = {Daniel Wichs and
                  Yishay Mansour},
  title        = {Geometric median in nearly linear time},
  booktitle    = {48th Annual {ACM} {SIGACT} Symposium on Theory
                  of Computing ({STOC} 2016)},
  pages        = {9--21},
  publisher    = {{ACM}},
  year         = {2016},
  *url          = {https://doi.org/10.1145/2897518.2897647},
  doi          = {10.1145/2897518.2897647},
  timestamp    = {Tue, 06 Nov 2018 11:07:06 +0100},
  biburl       = {https://dblp.org/rec/conf/stoc/CohenLMPS16.bib},
  bibsource    = {dblp computer science bibliography, https://dblp.org}
}

@article{overmars1981,
  author       = {Mark H. Overmars and
                  Jan van Leeuwen},
  title        = {Maintenance of Configurations in the Plane},
  journal      = {Journal of Computer and System Sciences},
  volume       = {23},
  number       = {2},
  pages        = {166--204},
  year         = {1981},
  *url          = {https://doi.org/10.1016/0022-0000(81)90012-X},
  doi          = {10.1016/0022-0000(81)90012-X},
  timestamp    = {Tue, 16 Feb 2021 14:03:43 +0100},
  biburl       = {https://dblp.org/rec/journals/jcss/OvermarsL81.bib},
  bibsource    = {dblp computer science bibliography, https://dblp.org}
}

@inproceedings{chan20dynam_gener_closes_pair,
  author       = {Timothy M. Chan},
  *editor       = {Martin Farach{-}Colton and
                  Inge Li G{\o}rtz},
  title        = {{Dynamic generalized closest pair: Revisiting Eppstein's technique}},
  booktitle    = {3rd Symposium on Simplicity in Algorithms ({SOSA} 2020)},
  pages        = {33--37},
  publisher    = {{SIAM}},
  year         = 2020,
  *url          = {https://doi.org/10.1137/1.9781611976014.6},
  doi          = {10.1137/1.9781611976014.6},
  timestamp    = {Tue, 09 Mar 2021 20:52:21 +0100},
  biburl       = {https://dblp.org/rec/conf/soda/Chan20.bib},
  bibsource    = {dblp computer science bibliography, https://dblp.org}
}

@book{bkos-cgaa-08,
  author =    "de Berg, Mark and Cheong, Otfried and van Kreveld, Marc and Overmars, Mark",
  title =    {{Computational Geometry: Algorithms and Applications}},
  edition =    "3rd",
  publisher =    "Springer",
  address =    "Berlin",
  year =    {2008}
}

@article{andersson99gener_balan_trees,
  author    = {Arne Andersson},
  title     = {General Balanced Trees},
  journal   = {Journal of Algorithms},
  volume    = 30,
  number    = 1,
  pages     = {1--18},
  year      = 1999,
  *url       = {https://doi.org/10.1006/jagm.1998.0967},
  doi       = {10.1006/jagm.1998.0967},
  timestamp = {Sun, 28 May 2017 13:24:57 +0200},
  biburl    = {https://dblp.org/rec/journals/jal/Andersson99.bib},
  bibsource = {dblp computer science bibliography, https://dblp.org}
}

@article{edelsbrunner86const_arran_lines_hyper_applic,
  author       = {Herbert Edelsbrunner and
                  Joseph O'Rourke and
                  Raimund Seidel},
  title        = {Constructing arrangements of lines and hyperplanes with applications},
  journal      = {{SIAM} Journal on Computing},
  volume       = 15,
  number       = 2,
  pages        = {341--363},
  year         = 1986,
  *url          = {https://doi.org/10.1137/0215024},
  doi          = {10.1137/0215024},
  timestamp    = {Wed, 14 Nov 2018 10:45:07 +0100},
  biburl       = {https://dblp.org/rec/journals/siamcomp/EdelsbrunnerOS86.bib},
  bibsource    = {dblp computer science bibliography, https://dblp.org}
}

@article{kuhn1955hungarian,
  title={The {H}ungarian method for the assignment problem},
  author={Kuhn, Harold W.},
  journal={Naval Research Logistics Quarterly},
  volume={2},
  number={1--2},
  pages={83--97},
  year={1955},
  doi={10.1002/nav.3800020109},
  publisher={Wiley Online Library}
}

@book{preparata2012computational,
  title={{Computational Geometry: An Introduction}},
  author={Preparata, Franco P. and Shamos, Michael I.},
  year={2012},
  publisher={Springer Science \& Business Media}
}

@article{edmonds1972theoretical,
  title={Theoretical improvements in algorithmic efficiency for network flow problems},
  author={Edmonds, Jack and Karp, Richard M.},
  journal={Journal of the ACM},
  volume={19},
  number={2},
  pages={248--264},
  year={1972},
  publisher={ACM New York, NY, USA}
}

@article{driscoll89makin,
title = {Making data structures persistent},
journal = {Journal of Computer and System Sciences},
volume = 38,
number = 1,
pages = {86-124},
year = 1989,
issn = {0022-0000},
doi = {https://doi.org/10.1016/0022-0000(89)90034-2},
*url = {https://www.sciencedirect.com/science/article/pii/0022000089900342},
author = {James R. Driscoll and Neil Sarnak and Daniel D. Sleator and Robert E. Tarjan},
}

\end{document}